\documentclass[journal]{IEEEtran}

\usepackage[english]{babel}

\usepackage{ifpdf}

\usepackage{cite} % Orders citations.
\usepackage{url}
\usepackage{hyperref}

\ifCLASSINFOpdf
	\usepackage[pdftex]{graphicx}
	\graphicspath{{./figures/}}
\else
	\usepackage[dvips]{graphicx}
	\graphicspath{./figures/}
\fi
\usepackage{color}
\usepackage{pgf, tikz, pgfplots}
\usetikzlibrary{shapes, arrows, automata}
\usepackage[cmex10]{amsmath}
\usepackage{amsfonts, amssymb, amsthm}
\usepackage{mathrsfs}
\usepackage{algorithm,algpseudocode}
\algnewcommand{\LeftComment}[1]{\Statex \(\triangleright\) #1}

\usepackage{enumerate}
\usepackage{multirow}
\usepackage{rotating}
\usepackage{subcaption}
\usepackage{needspace}

\input{mySymbol.sty}

\definecolor{pennpurple}{rgb}{0.537,0,0.510} % (Purple: RGB=(74,0,66); Light purple: RGB=(137,0,130))

\def\MSE{\mathrm{MSE}}
\def\Tr{\mathsf{T}}
\def\Hr{\mathsf{H}}

\def\RES{\text{RES}}
\def\tbarA{\check{\bbA}}

\newtheorem{assumption}{\hspace{0pt}\bf Assumption}
\newtheorem{lemma}{\hspace{0pt}\bf Lemma}
\newtheorem{proposition}{\hspace{0pt}\bf Proposition}

\newtheorem{theorem}{\hspace{0pt}\bf Theorem}
\newtheorem{corollary}{\hspace{0pt}\bf Corollary}

\newtheorem{definition}{\hspace{0pt}\bf Definition}

\begin{document}

\title{Controllability of Bandlimited Graph Processes Over Random Time Varying Graphs}

\author{Fer\hspace{0.01cm}nando~Gama, Elvin~Isufi, Alejandro~Ribeiro and Geert~Leus%,~\IEEEmembership{Student~Member,~IEEE,}
        %and~Alejandro~Ribeiro%,~\IEEEmembership{Member,~IEEE}% <-this % stops a space
\thanks{Part of this works is presented in \cite{Gama18-Control}. Work in this paper is supported by NSF CCF 1717120, ARO W911NF1710438, ARL DCIST CRA W911NF-17-2-0181, ISTC-WAS and Intel DevCloud. F. Gama, and A. Ribeiro are with the Dept. of Electrical and Systems Eng., Univ. of Pennsylvania, USA. E. Isufi is with the Intelligent Systems Dept. and G. Leus is with the Dept. of Microelectronics, Delft Univ. of Technology, The Netherlands. E-mails: \{fgama, aribeiro\}@seas.upenn.edu, \{e.isufi-1, g.j.t.leus\}@tudelft.nl.
}
}

% Headers:
\markboth{IEEE TRANSACTIONS ON SIGNAL PROCESSING (ACCEPTED)}%
{Controllability of Bandlimited Graph Signals}

\maketitle

\begin{abstract}
Controllability of complex networks arises in many technological problems involving social, financial, road, communication, and smart grid networks. In many practical situations, the underlying topology might change randomly with time, due to link failures such as changing friendships, road blocks or sensor malfunctions. Thus, it leads to poorly controlled dynamics if randomness is not properly accounted for. We consider the problem of controlling the network state when the topology varies randomly with time. Our problem concerns target states that are \emph{bandlimited} over the graph; these are states that 
have nonzero frequency content only on a specific graph frequency band. We thus leverage graph signal processing and exploit the bandlimited model to drive the network state from a fixed set of control nodes. When controlling the state from a few nodes, we observe that spurious, out-of-band frequency content is created. Therefore, we focus on controlling the network state over the desired frequency band, and then use a graph filter to get rid of the unwanted frequency content. To account for the topological randomness, we develop the concept of controllability in the mean, which consists of driving the expected network state towards the target state. A detailed mean squared error analysis is performed to quantify the statistical deviation between the final controlled state on a particular graph realization and the actual target state. Finally, we propose different control strategies and evaluate their effectiveness on synthetic network models and social networks.
% EDICS:
% 66. NEG-SPGR Signal processing over graphs (filtering, transforms, etc) < NEG SIGNAL PROCESSING FOR NETWORKS AND GRAPHS
\end{abstract}


\IEEEpeerreviewmaketitle

%!TEX root = control.tex

%%%%%%%%%%%%%%%%%%%%%%%%%%%%%%
%%% SECTION : Introduction %%%	sec:intro
%%%%%%%%%%%%%%%%%%%%%%%%%%%%%%

\section{Introduction} \label{sec:intro}

The controllability of complex networks plays a fundamental role in our understanding of natural and technological systems. Relevant examples involve the control of social, biological, financial, road, communication, and smart grid networks. Different works have highlighted the importance of the network structure when controlling a system evolving on top of that network \cite{Strogatz01-Exploring, Newman06-Structure, Lombardi07-Controllability}. Other works controlled said system from a few control or \emph{driving} nodes \cite{Liu11-Controllability, Yuan13-Exact, Pasqualetti14-Limitations}. As an illustrative example, consider the Zachary's Karate club social network in Figure~\ref{fig:karateClub} \cite{Zachary77-KarateClub}. The network state may be an opinion profile (e.g., members thoughts on a topic) and controlling, or \emph{driving,} the network amounts to shaping those opinions towards a \emph{desired} or \emph{target} state. The social relationships between members affect the ability to control the opinion profile and the objective is to sway all members opinions from a few influencing members (the \emph{driving} nodes).

While providing seminal contributions on network controllability, the works in \cite{Strogatz01-Exploring, Newman06-Structure, Lombardi07-Controllability, Liu11-Controllability, Yuan13-Exact, Pasqualetti14-Limitations} ignore the coupling between the underlying topology and the target state (a.k.a. \textit{the graph signal}). Recent evidence from graph signal processing (GSP) \cite{Taubin00-Meshes, Shuman13-SPG, Sandryhaila14-BigDataSPG} has shown that this coupling can bring substantial benefits in graph signal sampling \cite{Chen15-Sampling}, interpolation \cite{Narang13-Interpolation, Gama18-Sketching}, adaptive reconstruction \cite{DiLorenzo18-Adaptive}, and observability of diffusion processes \cite{Isufi18-Observability}. A common point that unifies \cite{Shuman13-SPG, Sandryhaila14-BigDataSPG, Chen15-Sampling, Narang13-Interpolation, Gama18-Sketching, DiLorenzo18-Adaptive, Isufi18-Observability} is the so-called graph Fourier transform (GFT). The GFT expands the network state onto an orthonormal basis related to the underlying topology ---formed by the eigenvectors of a matrix that represents the network, such as the adjacency or the Laplacian matrix--- akin to how the discrete Fourier transform expands a time signal on the complex exponential orthonormal basis. The GFT basis vectors can be linked to different variability modes across the graph through the concept of total variation \cite{Shuman13-SPG, Sandryhaila14-Freq}; hence, named the \emph{graph oscillating modes}.

A particular class of graph signals is that of \emph{bandlimited graph signals}, i.e., signals that can be expressed by only a few graph oscillating modes. The number of active modes forms the graph signal \emph{bandwidth}. Likewise for time signals, purely bandlimited graph signals rarely exist. But the bandlimitedness prior poses a powerful and parsimonious model to develop practical tools. This prior is exploited in GSP for sampling \cite{Shuman13-SPG, Sandryhaila14-BigDataSPG, Ortega18-GSP}, where signals arising in economic \cite{Marques16-Sampling}, handwritten digits \cite{Chen15-Sampling, Gama18-Sketching} or brain functional imaging \cite{Huang16-Brain}, are approximately bandlimited. In the Zachary's Karate club example, a bandlimited network state corresponds to opinions polarized into a few clusters of like-minded members \cite{Chen15-Sampling}. Therefore, controlling the system to a bandlimited state implies imposing a similar opinion profile to members that influence each other; this influence is captured by the edge weight.

Network control towards a bandlimited graph signal is considered in \cite{Segarra16-Percolation}. The control signal is fed to a few driving nodes and is percolated through the graph until the target state is reached. The authors determined the trade-off between the control time and the number of driving nodes, provided conditions to reach any bandlimited state, and designed the control signals.

The work in \cite{Barbarossa16-Control} studied the challenge of controlling the network towards a bandlimited state with control signals of limited energy. The main result is the trade-off between the number of driving nodes and the control signal energy. But conditions to drive the network to any bandlimited state were not derived in this limited energy setting.

Along these lines, and in parallel with the shorter version of this paper \cite{Gama18-Control}, the work in \cite{Bazerque17-Control} used GSP to formulate the linear quadratic controller as an autoregressive moving average (ARMA) graph filter \cite{Isufi17-ARMA}. This ARMA formulation is used to control independently each graph mode to the desired state. But these control strategies require all nodes to act as driving nodes.

\begin{figure}
    \centering
    \includegraphics[width = 0.52\columnwidth]{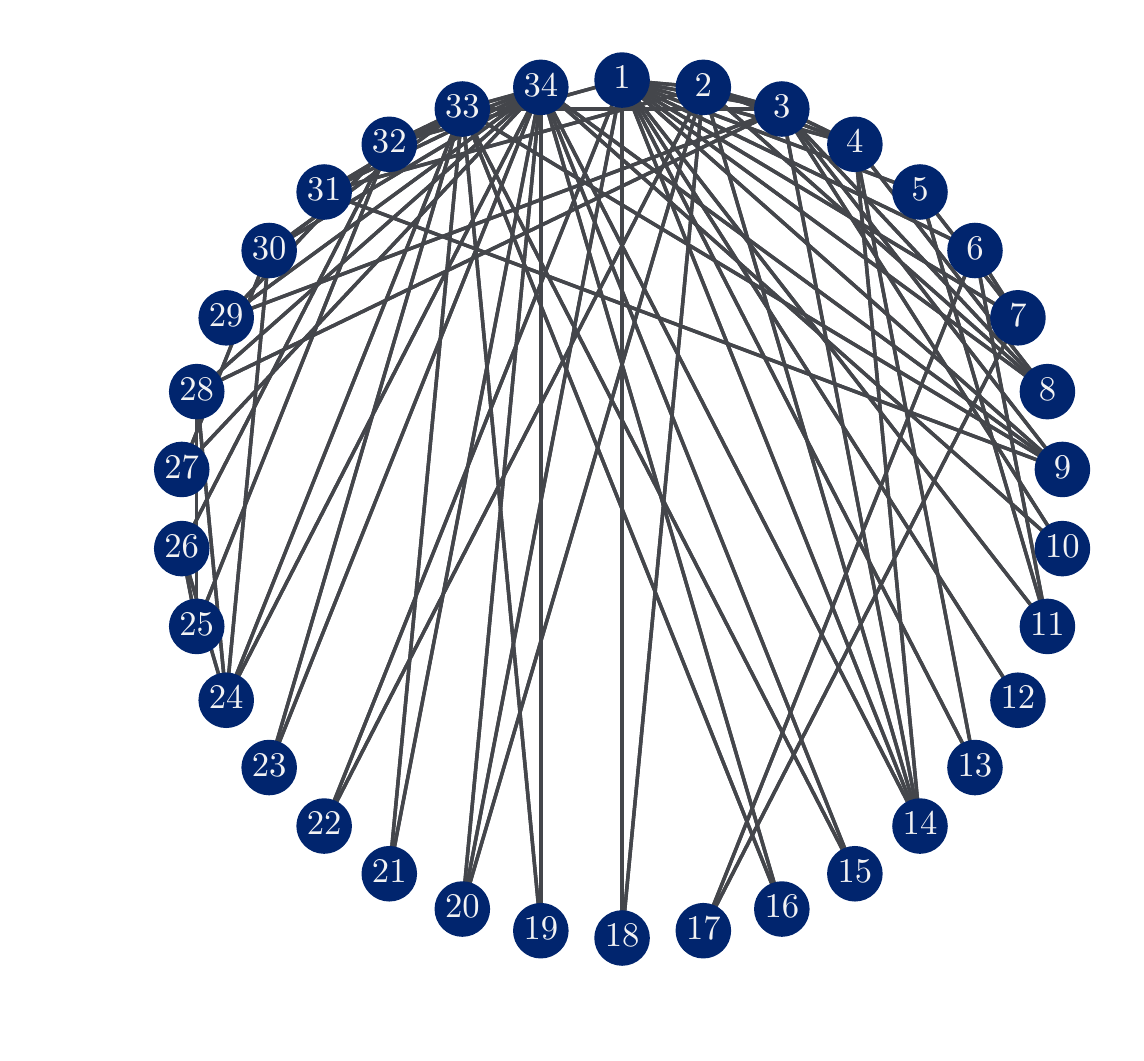}
    \caption{Zachary's Karate club network. The nodes represent club members and the edges capture social relationships between the members.
    }
    \label{fig:karateClub}
\end{figure}

Altogether, the above works concern network control over time-invariant topologies. But in practice the network structure may change randomly over time due to link losses, or nodes that disappear with a given probability. This might be the case of: $i)$ members that are not present a given day in the Zachary's Karate club; $ii)$ communication links that are random in sensor networks; $iii)$ power lines and buses that go down in smart grids due to local failures. In these situations, the graph structure is random and can be characterized by an expected graph with a variance around this expectation. Network controllability becomes, therefore, challenging since some connections cannot be exploited; hence, the derivations obtained in the deterministic setting may lead to a completely different state.

Motivated by the above observations, we study the possibilities to perform open loop control over random time varying networks. By exploring GSP tools as in \cite{Segarra16-Percolation, Barbarossa16-Control, Bazerque17-Control}, we propose a framework that accounts for the graph randomness in the analysis. Our main contributions are:

\begin{enumerate}
\item We study the problem of controlling dynamics over deterministic networks from a few driving nodes (Section~\ref{sec_sparse}). We provide conditions on network controllability that relate the minimum number of driving nodes to the target control bandwidth and the control time. This result encompasses the three strategies proposed in \cite{Segarra16-Percolation}.

\item We formulate the problem of controlling dynamics over random time varying networks from a few driving nodes. (Section~\ref{sec_random}). We develop the concept of controllability in the mean to drive the expected state towards a graph signal with a target bandwidth content on the expected graph. We also extend the conditions on network controllability from the deterministic to the stochastic setting.

\item We perform a mean squared error analysis to quantify the statistical deviation of the controlled state from the target state on a particular graph realization (Section~\ref{subsec_mse}). This analysis illustrates the role of the random graph model, the target state bandwidth, and the control signal.

\item We propose two control strategies to drive the expected state to the desired bandlimited state with the minimum mean squared error (Section~\ref{sec_strategies}).

\item We corroborate the developed framework and study its performance on synthetic (Erd\H{o}s-R\'{e}nyi and geometric graphs) and real-world social networks (Facebook sub-graph and Zachary's Karate club) (Section~\ref{sec_sims}).
\end{enumerate}

To the best of our knowledge, this is the first contribution that approaches network controllability over random time varying topologies. We remark that while this work relies on the bandlimited prior to select the driving nodes, the reader might find of interest other parsimonious models \cite{Chen16-Piecewise, Ortega18-GSP, Jung19-Localized} for such a task.

The remaining part of the paper proceeds as follows: Section~\ref{sec_prelim} sets down the preliminary concepts and Section~\ref{sec_sparse} contains our formulation of network controllability on deterministic graphs. Section~\ref{sec_random} formulates the controllability on random graphs. Section~\ref{sec_strategies} develops the control strategies, while Section~\ref{sec_sims} presents the numerical experiments. Finally, Section~\ref{sec_conclusions} provides the concluding remarks. The proofs are collected in the appendix.

\vspace{0.1cm} \noindent \textbf{Notation.} Normal letters $a$ (or $A$) denote scalars, bold lowercase letters $\bba$ vectors, and bold uppercase letters $\bbA$ matrices. The $i$th entry of a vector is $[\bba]_{i}$ while the $(i,j)$th entry of a matrix is $[\bbA]_{i,j}$. Superscripts $\Tr$ and $\Hr$ denote the transpose and the Hermitian, respectively. The $N \times 1$ null vector is $\bbzero_N$, the $N \times 1$ vector of all ones is $\bbone_N$ and the $N \times N$ identity matrix is $\bbI_N$. The diagonal operator denoted as $\diag(\cdot)$ is defined such that $\bba = \diag(\bbA)$ with $[\bba]_{i}=[\bbA]_{i,i}$, and $\bbA = \diag(\bba)$ is a diagonal matrix with vector $\bba$ on the main diagonal. The expectation operator is denoted as $\mbE[\cdot]$, the trace operator as $\tr(\cdot)$, the rank of $\bbA$ as $\rank(\bbA)$, the Kronecker product as $\otimes$, and the element-wise Hadamard product as $\circ$. The $l_{p}$ vector and matrix norm is denoted as $\|\cdot\|_{p}$. The ceiling operator is denoted as $\lceil\cdot\rceil$ and the minimum and maximum operators as $\min\{\cdot\}$ and $\max\{\cdot\}$, respectively. If not otherwise stated, calligraphic letters $\ccalA$ indicate sets and the set cardinality is denoted as $|\ccalA|$.

%!TEX root = control.tex

%%%%%%%%%%%%%%%%%%%%%%%%%%%%%%%
%%% SECTION : Preliminaries %%%
%%%%%%%%%%%%%%%%%%%%%%%%%%%%%%%

\section{Diffusion Processes on Graphs} \label{sec_prelim}

In this work, we consider controlling a diffusion process on random time varying graphs towards a desired state. To achieve this, we model diffusion processes through graph signal processing. We introduce the basics of GSP in Section~\ref{subsec_GFT}, define the random time varying graph model in Section~\ref{subsec_randgraph}, and discuss diffusion processes in Section~\ref{subsec_diffusion}. 

%%%%%%%%%%%%%%%%%%%
%%% SUBSECTION
%%%	subsec_GFT
%%%%%%%%%%%%%%%%%%%

\subsection{Graph signal processing} \label{subsec_GFT}

Let $\ccalG=(\ccalV,\ccalE,\ccalW)$ denote a graph with $\ccalV = \{v_{1},\ldots,v_{N}\}$ the set of $N$ vertices, $\ccalE \subseteq \ccalV \times \ccalV$ the set of edges, and $\ccalW:\ccalE \to \reals_{+}$ a function that assigns positive weights to the edges. The graph serves as a mathematical representation of the network and its structure is captured by the graph shift operator (GSO) matrix $\bbS \in \reals^{N \times N}$. The $(i,j)$th element of $\bbS$, $[\bbS]_{i,j}$, is nonzero only if $i = j$ or if $(v_j,v_i) \in \ccalE$; so that $\bbS$ respects the sparsity of $\ccalG$. Standard choices for $\bbS$ are the weighted graph adjacency matrix $\bbW$ \cite{Sandryhaila13-DSPG, Sandryhaila14-Freq}, the graph Laplacian matrix $\bbL$  \cite{Shuman13-SPG}, or their respective generalizations {\cite[Chapter 8]{Godsil01-AlgebraicGraphTheory}}. We consider $\bbS$ admits an eigendecomposition $\bbS = \bbV \bbLambda \bbV^{\Hr}$, where $\bbV = [\bbv_{1},\ldots,\bbv_{N}] \in \mbC^{N \times N}$ collects the orthonormal eigenvectors and $\bbLambda = \diag(\lambda_{1},\ldots,\lambda_{N}) \in \mbC^{N \times N}$ contains the associated eigenvalues. This holds for all undirected graphs on which the graph Laplacian can be defined and also for the adjacency matrix of some directed graphs \cite{Sandryhaila14-Freq, Marques16-Sampling}.

A graph signal is a mapping from the vertex set to the field of real numbers, i.e., $x_{i} : v_{i} \to \reals$, for $v_{i} \in \ccalV$. An example of a graph signal is the opinion of members in Zachary's Karate club. We collect all node signals in the vector $\bbx \in \reals^{N}$ with $[\bbx]_{i}=x_{i}$ being the value of node $i$ \cite{Shuman13-SPG}.

The graph Fourier transform (GFT) is the projection of the graph signal $\bbx$ on the eigenbasis $\bbV$ and is denoted by $\tbx = \bbV^{\Hr} \bbx$ \cite{Taubin00-Meshes, Shuman13-SPG}. The elements $[\tbx]_{k}=\tdx_{k}$ denote the graph Fourier coefficients of $\bbx$, whereas the eigenvectors $\bbv_{k}$ form the basis of \emph{graph oscillating modes}. Likewise, the inverse GFT is $\bbx = \bbV \tbx$, i.e., it writes $\bbx$ as a linear combination of the graph oscillating modes weighed by the Fourier coefficients.

A graph signal is \emph{bandlimited} if it has only a few nonzero Fourier coefficients. Without loss of generality, assume the first $K$ elements of $\tbx$ are nonzero; so, we can write $\tbx = [\tbx_{K}^{\Tr},\bbzero_{N-K}^{\Tr}]^{\Tr}$ where $\tbx_{K} \in \mbC^{K}$ and $K \le N$. Then, $\bbx$ is written in the compact form

\begin{equation}
\label{eq.gft_comp}
\bbx =  \bbV_{K}\tbx_{K}
\end{equation}
where $\bbV_{K} \in \mbC^{N \times K}$ is the respective column-trimmed eigenvector matrix. The GFT $\tbx_{K}$ of $\tbx$ writes as $\tbx_{K} = \bbV_{K}^{\Hr} \bbx$. This representation connects the signal bandwidth with the sampling and reconstruction strategies as shown in \cite{Pesenson08-Sampling, Chen15-Sampling, Tsitsvero16-Uncertainty, DiLorenzo18-Adaptive, Gama18-Sketching}. We will also exploit bandlimitedness in Section~\ref{sec_random} to control the network from a few driving nodes.

%%%%%%%%%%%%%%%%%%%
%%% SUBSECTION
%%%	subsec_randgraph
%%%%%%%%%%%%%%%%%%%

\subsection{Random time varying graphs} \label{subsec_randgraph}

We consider the following random graph model.

%% DEF
% def_res
\begin{definition}[$\RES(p)$ graph model \cite{Isufi17-Random}] \label{def_res}
Given an underlying graph $\ccalG\!=\!(\ccalV,\ccalE)$, a random edge sampling (RES) graph realization $\ccalG_{t} \!=\! (\ccalV,\ccalE_{t})$ of $\ccalG$ consists of the same set of nodes $\ccalV$ and assumes the edge $(v_i, v_j) \in \ccalE$ is sampled at time $t$ (i.e., $(v_i,v_j) \in \ccalE_{t}$) with a probability $0 < p \leq 1$. The edges are sampled independently over both the graph and the temporal dimension and are mutually independent from the graph signal if the latter has a stochastic nature.
\end{definition}

In other words, the $\RES(p)$ model states that the realization $\ccalG_{t} = (\ccalV,\ccalE_{t})$ is drawn from the underlying graph $\ccalG=(\ccalV,\ccalE)$, where the instantaneous edge set $\ccalE_t \subseteq \ccalE$ is generated via an independent Bernoulli process with probability $p$. Let us from now on denote with $\bbW$ and $\bbD = \diag(\bbW \bbone_{N})$ the adjacency and degree matrix of $\ccalG$, respectively. If the graph is undirected, we will also consider the unnormalized graph Laplacian matrix $\bbL = \bbD - \bbW$. To ease the exposition, denote with $\bbW_t$, $\bbD_t$, and $\bbL_t$ the respective matrices of the instantaneous graph $\ccalG_t$ and with $\barbW = \mbE[\bbW_t]$, $\barbD = \mbE[\bbD_{t}]$, and $\barbL = \mbE[\bbL_{t}]$ those of the expected graph $\bar{\ccalG}$. Under the $\RES(p)$ model, it holds that $\barbW = p\bbW$, $\barbD = p\bbD$, and $\barbL = p\bbL$. 

We assume the following.
\begin{assumption} \label{ass_gso}
The GSO of the underlying graph $\ccalG$ has an upper bounded spectral norm $\|\bbS\|_{2} \leq \varrho$ for some $\varrho < \infty$.
\end{assumption}
\noindent This assumption is generally met in practice and implies the graphs of interest have finite dimension and edge weights.

We also remark that more complex models than the $\RES(p)$ can be found in literature. In particular, many of these results can be readily extended to the model in which each edge $(v_i, v_j)$ is sampled independently with a different probability $p_{ij}$. But for clarity of exposition, we will focus only on the $\RES(p)$ model.

%%%%%%%%%%%%%%%%%%%
%%% SUBSECTION
%%%	subsec_diffusion
%%%%%%%%%%%%%%%%%%%

\subsection{Diffusion on graphs from a GSP perspective} \label{subsec_diffusion}

The continuous-time diffusion of a signal $\bbx_0$ on a graph $\ccalG$ with Laplacian matrix $\bbL$ is described by the differential equation \cite{Kondor02-Diffusion, Thanou17-Heat}
% EQN: eqn_cont_diffeq
\begin{equation} \label{eqn_cont_diffeq}
	\frac{d \bbx(s)}{ds} = - \bbL \bbx(s), 
		\qquad
	\bbx(0)=\bbx_{0}.
\end{equation}
This equation can be discretized as \cite{Olfati07-Consensus}
% eqn_discrete_sol
\begin{equation} \label{eqn_discrete_sol}
	\bbx_{t} = \bbA \bbx_{t-1}, 
		\qquad 
	\bbA = \bbI - \epsilon \bbL, 
		\qquad 
	t \in \mbN
\end{equation}
which is stable if $0 < \epsilon \le 1/\|\bbL\|_{2}$. Alternatively, a diffusion process on a graph can be interpreted as the discrete-time shift of $\bbx_0$ through the graph edges \cite{Sandryhaila13-DSPG, Gama19-GLLN}
% EQN: eqn_discrete_diff
\begin{equation} \label{eqn_discrete_diff}
	\bbx_{t} = \bbA \bbx_{t-1}, 
		\qquad 
	\bbA = \bbW, 
		\qquad 
	t \in \mbN.
\end{equation}
%

%Models \eqref{eqn_discrete_sol} and \eqref{eqn_discrete_diff} \blue{are two examples of widespread use for graph diffusion processes. 
Model \eqref{eqn_discrete_sol} is used when the process is defined over a continuous space that has been discretized, usually in the form of a mesh, as for heat diffusion processes \cite{Thanou17-Heat}. Model \eqref{eqn_discrete_diff} is employed when the underlying support is naturally a graph, as for sensor network communications \cite{Gama19-GLLN}. In essence, these are two examples of processes that describe the network state evolution by $\bbx_{t} = \bbA \bbx_{t-1}$ and relate this state to the underlying time-invariant topology (the transition matrix $\bbA$ depends on the shift operator). In this paper, we consider the more general case of random time varying topologies, i.e. $\bbx_{t} = \bbA_{t-1} \bbx_{t-1}$, and we abstract the relationship between the transition matrix and the underlying topology as follows.
\begin{assumption} \label{ass_transition_matrix}
Let $\bbA_{t}$ be the time varying transition matrix of a diffusion process over a random time varying graph $\ccalG_{t} \in \RES(p)$. Then, $\mbE[\bbA_{t}]$ and $\bbS$ share the same eigenvectors.
\end{assumption}
\noindent That is, we consider diffusions on random graphs such that the eigenvectors of the expected transition matrix and the underlying GSO coincide. The following lemma shows this is the case for the diffusion models \eqref{eqn_discrete_sol} and \eqref{eqn_discrete_diff} on $\RES(p)$ graph realizations.
\begin{lemma} \label{l_valid_models}
Let $\ccalG$ be a graph satisfying Assumption~\ref{ass_gso} and let $\ccalG_{t}$ be a $\text{RES}(p)$ realization of it. For the diffusion models
\begin{enumerate}[(i)]
\item $\bbS = \bbL$ and $\bbA_{t} = \bbI - \epsilon \bbL_{t}$ [cf. \eqref{eqn_discrete_sol}],
\item $\bbS = \bbW$ and $\bbA_{t} = \bbW_{t}$ [cf. \eqref{eqn_discrete_diff}],
\end{enumerate}
Assumption~\ref{ass_transition_matrix} and $\|\bbA_{t}\|_2 \leq \varrho$ hold.
\end{lemma}
\noindent Other models that satisfy these conditions are the wave equation on graphs and graph-based ARMA models, see \cite{Isufi18-Observability}.

%!TEX root = control.tex

%%%%%%%%%%%%%%%%%%%%%%%%%%%%%%%%
%%% SECTION : Sparse Control %%%
%%%%%%%%%%%%%%%%%%%%%%%%%%%%%%%%

\section{Controllability on Deterministic Graphs} \label{sec_sparse}

Consider the $N$-state linear system
\begin{equation} \label{eqn_control_sys}
	\bbx_{t} = \bbA \bbx_{t-1} + \bbB \bbu_{t-1}
\end{equation}
where $\bbx_t \in \reals^N$ denotes the state value on all nodes at time $t$, $\bbu_t \in \reals^M$ is the control signal injected on $M \le N$ nodes, and $\bbA \in \reals^{N \times N}$ and $\bbB \in \reals^{N\times M}$ are the transition and control input matrix, respectively. The relationship between the network state $\bbx_{t}$ and the underlying topology is captured in \eqref{eqn_control_sys} through the transition matrix $\bbA$; it shares the eigenvectors with $\bbS$ and this is the case for models \eqref{eqn_discrete_sol} and \eqref{eqn_discrete_diff}.

System \eqref{eqn_control_sys} is controllable if and only if the controllability matrix
% EQN: eq.cont_mat
\begin{equation} \label{eq.cont_mat}
\bbOmega = [\bbB, \bbA\bbB, \ldots, \bbA^{T-1}\bbB]
\end{equation}
has full rank $N$ \cite[Section 6.2.1]{Kailath80-LinearSystems}. While full rank of $\bbOmega$ guarantees the convergence of $\bbx_t$ to any target signal $\bbx^*$, we focus on controlling the network state towards a bandlimited graph signal $\bbx^{\ast}=\bbV_{K} \tbx_{K}^{\ast}$. Here, $\tbx_{K}^{\ast} \in \mbC^{K}$ determines the desired frequency response over the $K$ frequencies of interest; the \emph{target bandwidth}. We thus define the bandwidth controllability as follows.
\begin{definition}[Bandwidth controllability] \label{def:BWcontrol}
An $N$-state system on a graph is bandwidth controllable from $M \le N$ nodes if, for any initial state $\bbx_0$ and some final time $T$, there exists a sequence of control signals $\{\bbu_{t} \in \reals^{M}, t=0,1,\ldots,T-1\}$ acting on a fixed set of $M$ nodes that drive the network state to a value $\bbx^*$ with any frequency content $\tbx_{K}^{\ast} = \bbV_{K}^{\Hr} \bbx^{\ast}$ over the $K \leq N$ target bandwidth.
\end{definition}

Lead by the promising results of bandlimited graph signal reconstruction from samples on a few nodes \cite{Pesenson08-Sampling, Chen15-Sampling, Tsitsvero16-Uncertainty, DiLorenzo18-Adaptive, Gama18-Sketching}, we aim to control $\bbx_t$ through a fixed, time-invariant, set of nodes $\ccalS$ of cardinality $|\ccalS| = M \le N$. Let then $\bbB = \bbC^\Tr$ denote a binary matrix that selects these nodes. More formally, $\bbC$ belongs to the combinatorial set
%
% EQN: eqn_selection_set
\begin{equation} \label{eqn_selection_set}
	\ccalC_{M,N} = 
		\left\{ 
			\bbC \in \{0,1\}^{M \times N} : 
				\bbC \bbone_{N} = \bbone_{M} , 
				\bbC^{\Tr} \bbone_{M} \leq  \bbone_{N} 
		\right\}
\end{equation}
that selects $M$ out of $N$ different nodes and the ordering relation $\leq$ among vectors stands for the elementwise partial ordering \cite[Example 2.23]{Boyd04-Convex}. Observe that $\bbC \bbC^{\Tr} = \bbI_{M}$ and $\bbC^{\Tr} \bbC = \diag(\bbc)$ with $\bbc \in \{0,1\}^{N}$, such that $[\bbc]_{i}=1$ if and only if node $v_{i}$ belongs to $\ccalS$.

With this in place, we write the linear system on graphs \eqref{eqn_control_sys} in the GFT domain as
% EQN: eq.sys_gft
\begin{equation}\label{eq.sys_gft}
% SPLIT
\begin{split}
	\tbx_{t} 
		&= \bbV^\Hr\bbA\bbV\tbx_{t-1} + \bbV^\Hr\bbC^\Tr\bbu_{t-1}\\
		&\triangleq \tbA\tbx_{t-1} + \bbV^\Hr\bbC^\Tr\bbu_{t-1}
\end{split}
\end{equation}
where $\tbA = \bbV^\Hr\bbA\bbV$ is a diagonal matrix containing the eigenvalues of $\bbA$ [cf. Assumption~\ref{ass_transition_matrix}]. For convenience, we write $\tbA = \text{diag}(\bba) \in \mbC^{N \times N}$ with $\bba \in \mbC^{N}$ the vector containing the eigenvalues of $\tbA$, known also as the spectral response of $\bbA$. Then, by splitting \eqref{eq.sys_gft} into the $K$ frequencies of interest we have
% EQN: eq.gft_splitSys
\begin{equation}\label{eq.gft_splitSys}
 \begin{bmatrix}
	\tbx_{t,K} \\
	\tbx_{t,N-K}
\end{bmatrix}
	 =
\begin{bmatrix}
	\diag(\bba_{K}) \tbx_{t-1,K} \\
	\diag(\bba_{N-K}) \tbx_{t-1,N-K}
\end{bmatrix} 
+
\begin{bmatrix}
	\bbV_{K}^{\Hr}\bbC^{\Tr} \bbu_{t-1} \\
	\bbV_{N-K}^{\Hr}\bbC^{\Tr} \bbu_{t-1}
\end{bmatrix}
\end{equation}
where $\tbx_t = [\tbx_{t,K}^{\Tr},\tbx_{t,N-K}^{\Tr}]^{\Tr}$, $\bba = [\bba_{K}^{\Tr}, \bba_{N-K}^{\Tr}]^{\Tr}$ and $\bbV = [\bbV_K, \bbV_{N-K}]$.

For a system that is bandwidth controllable [cf. Def.~\ref{def:BWcontrol}], recursion \eqref{eq.gft_splitSys} can drive the network state $\bbx_{t}$ to any signal $\bbx_{T}$ that has a desired frequency response $\tbx_{K}^{\ast}$ over the target bandwidth $K$. Thus, we can focus on the $K$ frequencies of interest, determine the driving nodes through matrix $\bbC$, and design the control signals $\{\bbu_{t}\}$ such that the system
% EQN: eqn_control_sys_freq_k
\begin{equation} \label{eqn_control_sys_freq_k}
\tbx_{t,K} = 
\tbA_K \tbx_{t-1,K} 
+ \bbV_{K}^{\Hr} \bbC^{\Tr} \bbu_{t-1}.
\end{equation}
reaches $\tbx_{K}^{\ast}$ at time $T$, i.e. $\tbx_{T,K} = \tbx_{K}^{\ast}$. But we observe that focusing on these $K$ frequencies leads to a non-zero value also on the $N-K$ remaining frequencies. One way to suppress this undesired content is to force $\tbx_{t,N-K} = \bbzero_{N-K}$. This requires the design of a sampling matrix $\bbC$ that satisfies $\bbV_{N-K}^{\Hr} \bbC^{\Tr} = \bbzero_{(N-K) \times M}$ [cf. \eqref{eq.gft_splitSys}]. The latter might be infeasible or might severely constraint the selection of driving nodes. Therefore, to avoid the out-of-band frequency content, we use a frequency-selective graph filter $\bbH = \bbV_{K}\bbV_{K}^{\Hr}$ to force bandlimitedness on the final network state $\bbx^{\ast}=\bbH \bbx_{T}$. This filter can be implemented locally through a polynomial in the shift operator $\bbS$ with degree at most $N$ \cite{Segarra17-Linear}.

Hereinafter, the design variables are the sampling matrix $\bbC$ and the control signals $\bbu_t \in \reals^{M}$ for $t = 0,\ldots,T-1$. And, since the initial state $\bbx_{0}$ is considered known we fix, without loss of generality, $\bbx_{0}=\bbzero_{N}$ as the common practice in control literature \cite{Barbarossa16-Control}, {\cite[Section 2.3.2]{Kailath80-LinearSystems}, \cite[Section 2.1]{Lewis95-OptimalControl}}. With this set down, we claim our first contribution, which we will also exploit for controllability on random graphs in Section~\ref{sec_random}.
%%%% PROPOSITION:
%% prop.detSparCont
%
\begin{proposition}\label{prop.detSparCont}
Consider the linear system \eqref{eqn_control_sys_freq_k} describing a process over a deterministic graph $\ccalG$. A necessary condition to control the system in a finite time $T$ towards a target frequency content $\tbx^*_K$ over the $K$ frequencies of interest is to select $M \ge \lceil K/T \rceil$ driving nodes.
\end{proposition}
%
%%
%%%

Proposition~\ref{prop.detSparCont} provides a necessary condition on the minimum number of nodes to control system \eqref{eqn_control_sys_freq_k} in $T$ instants. It shows the trade-off between the cardinality of the sampling set $M$, the signal bandwidth $K$, and the control time $T$. Thus, for $T \ge K$ there is the potential to control the network by acting on a single node. But this is not sufficient since the system controllability is affected by the network topology \cite{Strogatz01-Exploring, Newman06-Structure, Lombardi07-Controllability}, i.e. the influence of the driving nodes on the frequency content, see also \cite{Chen15-Sampling, Marques16-Sampling, Segarra16-Percolation}. Hence, the driving nodes should be carefully picked to guarantee the controllability of \eqref{eqn_control_sys_freq_k}. In what follows, we show the relation of Proposition~\ref{prop.detSparCont} with \cite{Segarra16-Percolation} and \cite{Barbarossa16-Control}.

\vskip2.5mm
\textit{a) Relation with \cite{Segarra16-Percolation}.} Proposition~\ref{prop.detSparCont} encompasses under a single condition the three graph signal reconstruction strategies of \cite{Segarra16-Percolation}. In fact, $M \ge K$ and $T = 1$ covers the \emph{multiple node-single time} seeding strategy; $M = 1$ and $T \ge K$ is a necessary condition to control the signal for the \emph{single node-multiple time} seeding strategy; and $M \ge \lceil K/T \rceil$ covers the more involved \emph{multiple node-multiple time} seeding approach. This is expected, since graph signal reconstruction through percolation is a particular case of system \eqref{eqn_control_sys}.

\vskip2.5mm
\textit{b) Relation with \cite{Barbarossa16-Control}.} Differently from our approach, \cite{Barbarossa16-Control} focuses on designing the control signal $\bbu = [\bbu_{T-1}^\Tr,\ldots, \bbu_{1}^\Tr,\bbu_{0}^{\Tr}]^\Tr \in \reals^{MT \times 1}$ as a trade-off between sparsity in the vertex domain and signal energy. This problem writes as
% EQN: eq.brossa_prob
\begin{equation}\label{eq.brossa_prob}
% ALIGNED
\begin{aligned}
	& \underset{\bbu}{\text{minimize}}
	& & \|\bbu\|_2^2 + \gamma\|\bbu\|_0 \\
	& \text{subject to}
	& & \bbx_{t+1} = \bbA \bbx_{t} + \bbB \bbu_{t}\ , \ t = 0, \ldots, T-1 \\
	&&& \bbx_0 = \bb0_N, \quad \bbx_{T} = \bbx^*
\end{aligned}
\end{equation}
where the constant $\gamma$ trades the control signal energy  $\|\bbu\|_2^2$ with sparsity $\|\bbu\|_0$. Problem \eqref{eq.brossa_prob} yields a sparse control signal $\bbu$ across time, but the driving nodes are not necessarily fixed. In this regard, Proposition~\ref{prop.detSparCont} imposes a minimum dimension on $\bbu$ such that controllability is possible from a fixed set $\ccalS$ of driving nodes.

%!TEX root = control.tex

%%%%%%%%%%%%%%%%%%%%%%%%%%%%%%%%%%%%%%%%%%
%%% SECTION : Control on Random Graphs %%%
%%%%%%%%%%%%%%%%%%%%%%%%%%%%%%%%%%%%%%%%%%

\section{Controllability on Random Graphs} \label{sec_random}

When the network topology is time varying, system \eqref{eqn_control_sys} should change to reflect the time dependency in the transition matrix. When the time variation is random, the controllability of the network state should follow a statistical approach. We propose a statistical framework in this section, where in Section~\ref{subsec_mean} we develop the concept of controllability in the mean and in Section~\ref{subsec_mse} we perform the mean squared error analysis.

%%%%%%%%%%%%%%%%%%%
%%% SUBSECTION
%%%	subsec_mean
%%%%%%%%%%%%%%%%%%%

\subsection{Mean controllability} \label{subsec_mean}

The dynamics of a time varying system on random graphs are given by
% EQN: eqn_control_sys_t
\begin{equation} \label{eqn_control_sys_t}
	\bbx_{t} = \bbA_{t-1} \bbx_{t-1} + \bbC^{\Tr} \bbu_{t-1}
\end{equation}
where under the $\RES(p)$ model in Definition \ref{def_res}, $\{\bbA_{t}\}$ is a set of i.i.d. random matrices with $\mbE [\bbA_{t}] = \barbA$. The deterministic design variables are contained in the second term of \eqref{eqn_control_sys_t}, $\bbC^{\Tr} \bbu_{t-1}$. The state $\bbx_{t}$ depends on the random system matrices $\{\bbA_{\tau}\}_{\tau=0}^{t-1}$, which are independent from $\bbA_{t}$ and from the deterministic design variables $\bbC$ and $\{\bbu_{\tau}\}_{\tau=0}^{t-1}$. Note also that $\barbA$ and $\bbS$ share the same eigenvectors, i.e. $\barbA = \bbV \diag(\barba) \bbV^{\Hr}$; thus, it captures in statistics the relation between the underlying topology and state $\bbx_{t}$. We can then write the mean evolution of \eqref{eqn_control_sys_t} as
% EQN: eqn_control_sys_mean
\begin{align}
	\bbmu_{t}
		&= \barbA \bbmu_{t-1} + \bbC^{\Tr} \bbu_{t-1}
			\label{eqn_control_sys_mean}
\end{align}
where $\bbmu_{t}=\mbE[\bbx_{t}]$. System \eqref{eqn_control_sys_mean} is a deterministic and time-invariant system analogous to \eqref{eqn_control_sys}.  We develop the following controllability concept.

\begin{definition}[Bandwidth controllability in the mean] \label{def:BWcontrolRandom}
An $N$-state system on a random graph of the form in \eqref{eqn_control_sys_t} with mean evolution in \eqref{eqn_control_sys_mean} is bandwidth controllable in the mean from $M \le N$ nodes if, for any initial state $ \bbx_0 $ and some final time $T$, there exists a sequence of control signals $ \{\bbu_t, t = 0, \ldots, T-1 \} $ acting on a fixed set of $M$ nodes that drive the mean network state to a value $\bbx^*$ with any frequency content $\bbx_{K}^{\ast} = \bbV_{K}^{\Hr} \bbx^{\ast}$ over the $K \leq N$ target bandwidth.
\end{definition}

Our goal is to control the mean system to a desired bandlimited graph signal $\bbx^{\ast} = \bbV_{K} \tbx_{K}^{\ast}$ in a finite time $T$ from a few nodes. We do so by designing the input signals $\{\bbu_t, t = 0, \ldots, T-1 \}$ and the node driving set $\ccalS$ (through matrix $\bbC$). In analogy to Section~\ref{sec_sparse}, we focus on the $K$ frequencies of interest of the mean system [cf.~\eqref{eqn_control_sys_freq_k}]
% EQN: eqn_control_sys_freq_k_mean
\begin{equation} \label{eqn_control_sys_freq_k_mean}
	\tbmu_{t,K} = \tbarA_{K} \tbmu_{t-1,K} + \bbV_{K}^{\Hr} \bbC^{\Tr} \bbu_{t-1}
\end{equation}
and drive it to the desired frequency content $\tbx_{K}^{\ast}$ [cf. Def.~\ref{def:BWcontrolRandom}], with $\tbarA_{K} = \widetilde{\barbA}_{K} = \diag(\barba_{K}) \in \mbC^{K \times K}$ containing the eigenvalues of $\barbA$ which determine the spectral response of the system evolution on the expected graph. Then, we apply a (deterministic) linear filter $\bbH = \bbV_{K} \bbV_{K}^{\Hr}$ to keep only those $K$ desired frequencies such that the mean network state $\bbmu^{\ast} = \mbE[\bbH \bbx_{T}] = \bbH \bbmu_{T}=\bbx^{\ast}$ results in a bandlimited graph signal.

Similarly to Proposition~\ref{prop.detSparCont}, we claim the following.
%%%% PROPOSITION:
%% prop_nodes_control
%
\begin{proposition} \label{prop_nodes_control}
Consider the linear system \eqref{eqn_control_sys_t} describing a process over a sequence of $\RES(p)$ graphs $\ccalG_t$ with in-band mean evolution \eqref{eqn_control_sys_freq_k_mean}. A necessary condition to control the mean system in finite time $T$ towards a target frequency content $\tbx^*_K$ over the $K$ frequencies of interest is to select $M \ge \lceil K/T \rceil$ driving nodes.
\end{proposition}
Like Proposition~\ref{prop.detSparCont}, Proposition~\ref{prop_nodes_control} establishes a necessary condition to control a linear system, now, on random time varying graphs. As such, the same trade-off between the number of driving nodes $M$, the signal bandwidth $K$, and the control time $T$ applies here. The next corollary extends this result to a sufficient condition under some restrictions on the eigenvector basis.

%%%% COROLLARY:
%% cor_suff
%
\begin{corollary} \label{cor_suff}
Under the hypothesis of Proposition~\ref{prop_nodes_control}, if there exists a set of driving nodes $\ccalS$ built by $M$ nodes such that the corresponding $M$ rows of $\bbV_{K}$ are linearly independent vectors, then $M \ge K$ is a sufficient condition to control system \eqref{eqn_control_sys_t} in the mean.
\end{corollary}

\noindent An algorithm for finding such $M$ nodes is readily available in \cite[Algorithm 1]{Chen15-Sampling}.

%%%%%%%%%%%%%%%%%%%
%%% SUBSECTION
%%%	subsec_mse
%%%%%%%%%%%%%%%%%%%

\subsection{Mean squared error analysis} \label{subsec_mse}

In Section \ref{subsec_mean}, we discussed that system \eqref{eqn_control_sys_t} can be controlled \emph{in the mean}. Therefore, it is paramount also to quantify the mean squared error (MSE) of the controlled state to gain statistical insight into how \emph{close} the filtered final state on a specific graph realization $\bbH\bbx_{T}$ is to the actual desired signal $\bbx^{\ast}$. Towards this end, define $\bbPhi_{b,a} = \bbA_{b} \bbA_{b-1} \cdots \bbA_{a+1} \bbA_{a}$ as the \emph{state} transition matrix between time instants $b \geq a$. The following theorem determines the MSE.

%%%%%%%%
%%% THEOREM
%% thm_mse
%%%%%%
%%%
%%
\begin{theorem} \label{thm_mse}
Let Assumptions~\ref{ass_gso}~and~\ref{ass_transition_matrix} hold and let $\bbx_t$ be a graph process defined over a sequence of $\RES(p)$ graphs $\ccalG_t$ described by linear system \eqref{eqn_control_sys_t}. Given also a set of driving nodes $\ccalS$ characterized by selection matrix $\bbC$ and a set of control signals $\{\bbu_{t}\}_{t=0}^{T-1}$ with initial state $\bbx_0 = \bb0_N$. The MSE between the filtered signal $\bbH \bbx_{T}$ on a particular graph realization and the actual desired signal is
% EQN: eqn_mse_t
\begin{equation} \label{eqn_mse_t}
% SPLIT
\begin{split}
& \MSE(T) 
	= \mbE \left[ \left\| \bbH \bbx_{T} - \bbx^{\ast} \right\|_{2}^{2} \right]
		 \\
	& = \alpha 
		-2 \sum_{\tau=0}^{T-1} 
			\bbbeta_{\tau}^{\Tr} \bbC^{\Tr} \bbu_{\tau} 
		+ \sum_{\tau=0}^{T-1} \sum_{\tau'=0}^{T-1} 
			\bbu_{\tau}^{\Tr} \bbC \ \bbGamma_{\tau,\tau'} \ \bbC^{\Tr} \bbu_{\tau'}
\end{split}
\end{equation}
which is a quadratic form on $\bbC^{\Tr}\bbu_{\tau}$ with coefficients $\alpha = \| \bbx^{\ast} \|_{2}^{2} \in \reals$, $\bbbeta_{\tau} = (\barbA^{T-\tau-1})^{\Tr} \bbH^{\Tr} \bbx^{\ast} \in \reals^{N \times 1}$ and $\bbGamma_{\tau,\tau'} = \mbE [ \bbPhi_{T-1, \tau+1} \bbH^{\Tr} \bbH \bbPhi_{T-1,\tau'+1}] \in \reals^{N \times N}$.
\end{theorem}

Given $\bbC$ and $\{\bbu_{t}\}$, the MSE in \eqref{eqn_mse_t} holds for any system described by \eqref{eqn_control_sys_t}, irrespective of their controllability. The MSE in \eqref{eqn_mse_t} is a quadratic function in the design variables $\bbC^{\Tr}\bbu_{\tau}$ and the corresponding coefficients $\alpha$, $\bbbeta_{\tau}$, and $\bbGamma_{\tau,\tau'}$ depend on known quantities: the graph filter $\bbH$; the desired state $\bbx^{\ast} =\bbV_{K} \tbx_{K}^{\ast}$; and the statistics (first and second order moments) of the underlying support through $\barbA$ and $\bbGamma_{\tau,\tau'}$. This result highlights also the impact the driving nodes have on the overall performance and shows their connection with the underlying support and the target bandwidth. More precisely, the coefficient $\alpha$ provides the MSE floor if there is no control signal (i.e., given by the energy of the target state); $\bbbeta_{\tau}$ takes into account the similarity between the target signal and the controlled signal evolution over the mean graph; and $\bbGamma_{\tau,\tau'}$ accounts for the variability of the random graph. Finally, we note that the computation of $\bbGamma_{\tau,\tau'}$ in \eqref{eqn_mse_t} might be cumbersome for some transition matrices $\bbA_{t}$. We thus provide in the appendix two practical results that address this issue: first, we provide a general upper bound; second, we show how to exactly compute $\bbGamma_{\tau,\tau'}$ for undirected graphs for the diffusion models in Lemma~\ref{l_valid_models}.

%!TEX root = control.tex

%%%%%%%%%%%%%%%%%%%%%%%%%%%%%%%%%%%%
%%% SECTION : Control Strategies %%%
%%%%%%%%%%%%%%%%%%%%%%%%%%%%%%%%%%%%

\section{Control Strategies} \label{sec_strategies}

In this section, we propose two control strategies (i.e., find $\bbC$ and $\{\bbu_t\}$) for graph processes over random time varying graphs, where, depending on the scenario, one can be preferred over the other. In Section~\ref{subsec_unbiased} we propose an unbiased control strategy, while in Section~\ref{subsec_minmse_tau} we introduce a control strategy that leverages the bias-variance trade-off to minimize the MSE.

%%%%%%%%%%%%%%%%%%%
%%% SUBSECTION
%%%	subsec_unbiased_tau
%%%%%%%%%%%%%%%%%%%

\subsection{Unbiased controller} \label{subsec_unbiased}

%\textbf{Final time.} 
The mean state in \eqref{eqn_control_sys_freq_k_mean} for $t = T$ can be expanded as
% EQN: eqn_mu_t
\begin{equation} \label{eqn_mu_t}
	\tbmu_{T,K} = 
		\sum_{\tau=0}^{T-1} 
			\tbarA_{K}^{T-\tau-1} 
			\bbV_{K}^{\Hr} \bbC^{\Tr} \bbu_{\tau}
\end{equation}
where $\tbarA_{K} = \widetilde{\barbA}_{K} = \diag(\barba_{K}) \in \mbC^{K \times K}$ and $\tbmu_{0,K} = \bbzero_K$. For an unbiased controller, it must hold at final time $T$ $\tbmu_{T,K} = \tbx_{K}^{\ast}$. Combining \eqref{eqn_mu_t} and $\tbmu_{T,K} = \tbx_{K}^{\ast}$, we obtain
% EQN: eqn_linear_sys_unbiased
\begin{equation} \label{eqn_linear_sys_unbiased}
    \tbOmega \bbu = 
	\left[ 
		\bbI_{K}, 
		\tbarA_{K},
		\cdots,
		\tbarA_{K}^{T-1}
	\right]
	\left( \bbI_{T} \otimes \bbV_{K}^{\Hr} \bbC^{\Tr} \right)
	\bbu 
	= \tbx_{K}^{\ast}
\end{equation}
with in-band controllability matrix $\tbOmega \in \mbC^{K \times T M}$ and input vector is $\bbu = [\bbu_{T-1}^\Tr, \bbu_{T-2}^{\Tr},\ldots, \bbu_{1}^\Tr,\bbu_{0}^{\Tr}]^\Tr \reals^{MT \times 1}$. For $\tbOmega$ being of full rank $K$ (i.e. a controllable system), system \eqref{eqn_linear_sys_unbiased} has infinite solutions on $\bbu$. Also often exists more than one set of nodes that guarantees controllability. We then select the set of nodes and design the control signals to minimize the $\MSE(T)$, while guaranteeing the solution is unbiased. 

Let $\ccalC^{\ast}_{M,N} = \{\bbC \in \ccalC_{M,N} : \rank(\tbOmega) = K\}$ be the set of selection matrices that satisfy controllability. The optimal unbiased control strategy can be posed as
% ALGN: eqn_opt_unbiased_all
\begin{align}
\min_{\bbC \in \ccalC^{\ast}_{M,N}, \bbu \in \reals^{TM}}
& \MSE(T)
\label{eqn_opt_unbiased_all} \\
\textrm{s. t. }
& \tbOmega \bbu = \tbx_{K}^{\ast}
\nonumber \\
& \tbOmega = \left[ 
\bbI_{K}, 
\tbarA_{K},
\cdots,
\tbarA_{K}^{T-1}
\right]
\left( \bbI_{T} \otimes \bbV_{K}^{\Hr} \bbC^{\Tr} \right)
\nonumber
\end{align}
where $\MSE(T)$ is given in \eqref{eqn_mse_t}. Oftentimes, we are interested in controlling the system with minimum energy \cite{Pasqualetti14-Limitations, Barbarossa16-Control}. In such cases, the minimum energy control signal is \cite[Section 6.2]{Boyd04-Convex}
% EQN: eqn_optimal_wtau_unbiased
\begin{equation} \label{eqn_optimal_wtau_unbiased}
	\bbu^{\ast} 
		= \tbOmega^{\Hr} 
		\left[ \tbOmega\tbOmega^{\Hr} \right]^{-1} 
		\tbx_{K}^{\ast} .
\end{equation}
Then, within the minimum energy framework, we select the nodes that minimize the $\MSE(T)$ as follows
% ALGN: eqn_opt_unbiased
\begin{align}
\min_{\bbC \in \ccalC^{\ast}_{M,N}}
	& \MSE(T)
		\label{eqn_opt_unbiased} \\
\textrm{s. t. }
	& \bbu^{\ast} 
		= \tbOmega^{\Hr} 
		\left[ \tbOmega\tbOmega^{\Hr} \right]^{-1} 
		\tbx_{K}^{\ast},
		\nonumber \\
	& \tbOmega = \left[ 
		\bbI_{K}, 
		\tbarA_{K},
		\cdots,
		\tbarA_{K}^{T-1}
	\right]
	\left( \bbI_{T} \otimes \bbV_{K}^{\Hr} \bbC^{\Tr} \right).
		\nonumber
\end{align}

\begin{algorithm}[t]
    \caption{Constrained Greedy Approach.}
    \label{algm_greedy}
    
    \begin{algorithmic}[1]
        \Statex \textbf{Input:} $M$: number of samples, $T$: time horizon
        \Statex $\quad \tbx_{K}^{\ast}$: desired frequency response
        \Statex $\quad \bbV_{K}$: frequency basis vectors of active frequencies
        \Statex $\quad \tbarA_{K}$: GFT of transition matrix
        \Statex $\quad \MSE(\cdot)$: function to compute MSE \Comment{\emph{See \eqref{eqn_mse_t}}}
        \Statex \textbf{Output:} $\bbC$: selected nodes, $\{\bbu_{t}\}$: control signals
        \Statex
        \Procedure{greedy}{$M$, $T$, $\tbx_{K}^{\ast}$, $\bbV_{K}$, $\tbarA_{K}$, $\MSE(\cdot)$}
        \State Set $\ccalS = \emptyset$ \Comment{\emph{Selected nodes}}
        \State Set $\ccalR = \ccalV$ \Comment{\emph{Remaining nodes}}
        \State Set $\texttt{bestMSE} \leftarrow \infty$
        \For{$m = 1:M$}
            \State Set $\texttt{bestNode} \leftarrow \emptyset$
            \For{$n = 1:N - m +1$}
            \State Select $r_{n} \in \ccalR$ \Comment{\emph{Choose a remaining node}}
            \State Compute matrix $\bbC$ for $\ccalS \cup \{r_{n}\}$
            \State Compute matrix $\tbOmega$ \Comment{\emph{See \eqref{eqn_linear_sys_unbiased}}}
            \If{$\rank(\tbOmega) = \min\{K, T \min \{K,m\}\}$}
                \State Solve $\tbOmega \bbu = \tbx_{K}^{\ast}$ for $\bbu$
                \State Compute $\MSE(T)$ \Comment{\emph{See \eqref{eqn_mse_t}}}
                \If{$\MSE(T) < \texttt{bestMSE}$}
                    \State Set $\texttt{bestNode} \leftarrow r_{n}$
                    \State Set $\texttt{bestMSE} \leftarrow \MSE(T)$
                \EndIf
            \EndIf
            \EndFor
        \State Set $\ccalS \leftarrow \ccalS \cup \{\texttt{bestNode}\}$
        \State Set $\ccalR \leftarrow \ccalR \backslash \{\texttt{bestNode}\}$
        \EndFor
        \EndProcedure
    \end{algorithmic}
    
\end{algorithm}

%%% FIGURE %%%
%%
\begin{figure*}
    \captionsetup[subfigure]{justification=centering}
    \centering
    \begin{subfigure}{0.9\columnwidth}
        \centering
        \includegraphics[width=0.99\textwidth]{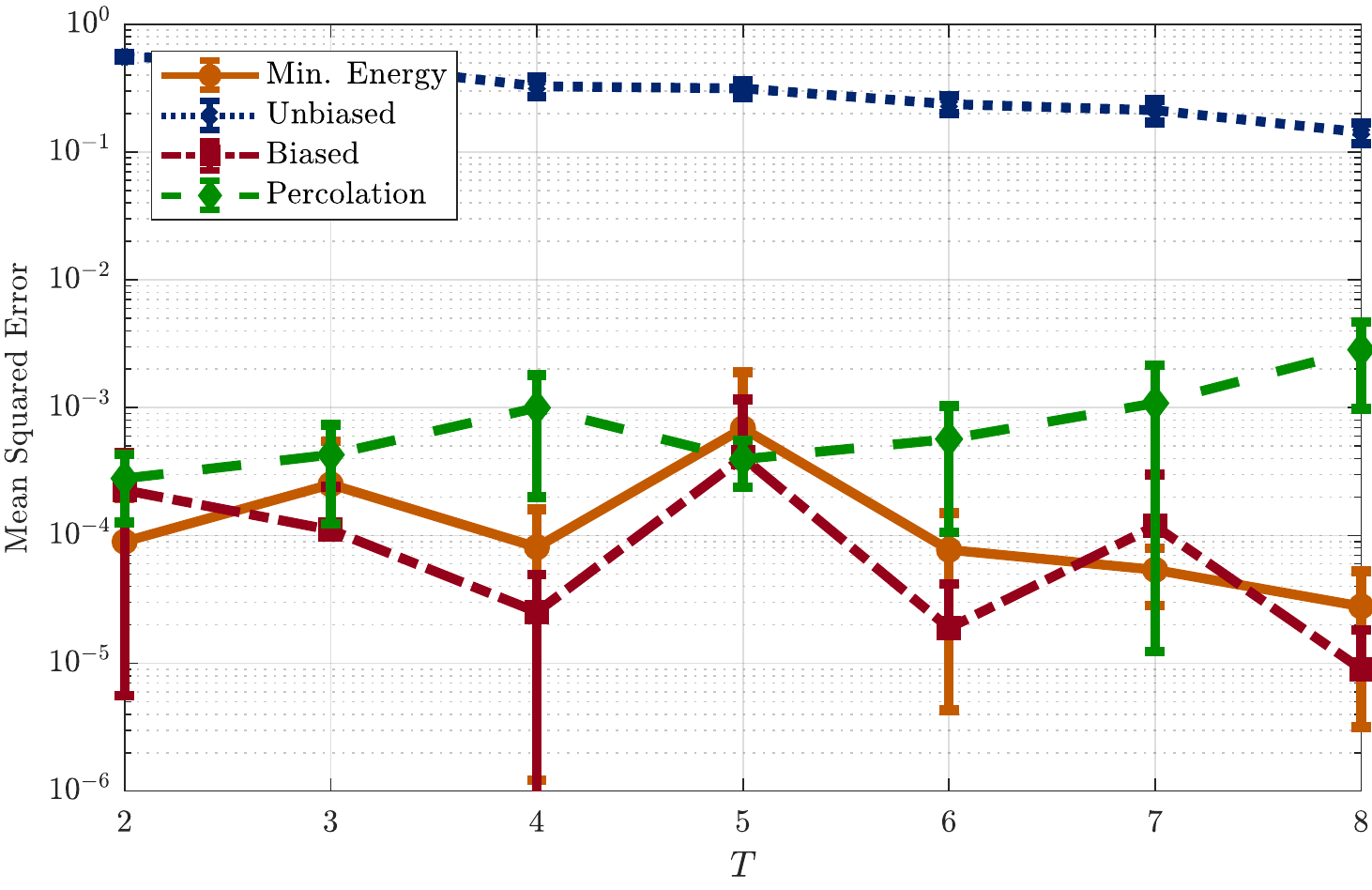}
        \caption{Sweeping time horizon $T$}
        \label{fixedT}
    \end{subfigure}
    \hfill
    \begin{subfigure}{0.9\columnwidth}
        \centering
        \includegraphics[width=0.99\textwidth]{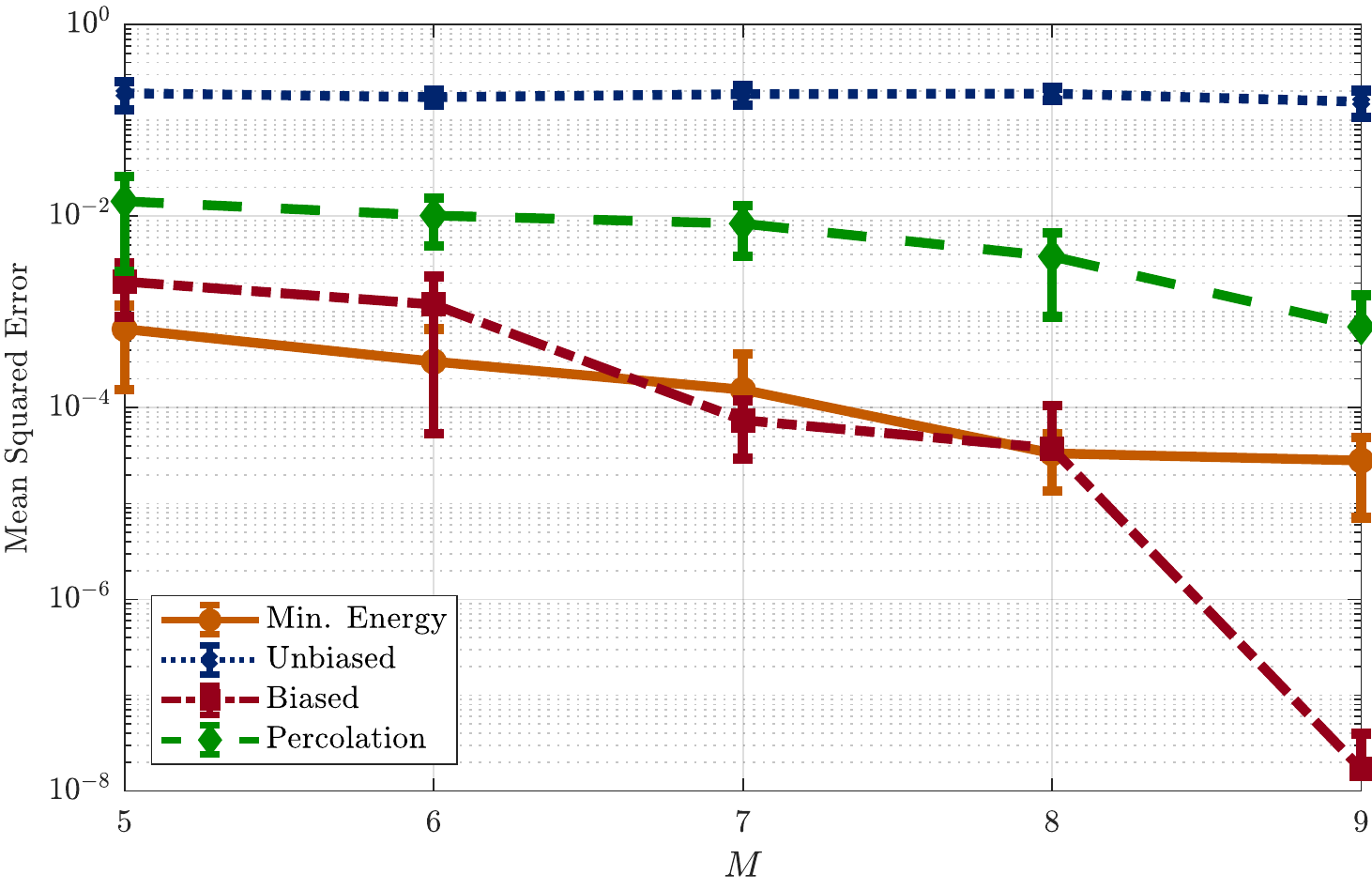}
        \caption{Sweeping number of selected nodes $M$}
        \label{fixedM}
    \end{subfigure}
    \caption{Baseline for time-invariant networks on geometric graphs. \subref{fixedT} Parameter sweeping simulation as a function of time horizon $T$. \subref{fixedM} Parameter sweeping simulation as a function of number of selected samples $M$. The biased controller has a performance similar to the Min. Energy and performs slightly better than Percolation. The unbiased controller lags behind in terms of MSE. The error bars indicate the estimated standard deviation from the $500$ graph realizations.}
    \label{fig:fixed}
\end{figure*}
%%
%%%%%%%%%%%%%%

Problem~\eqref{eqn_opt_unbiased} is non-convex due to the binary nature of the optimization variable $\bbC$. A heuristic solution is to follow a constrained greedy approach as described in Algorithm~\ref{algm_greedy}. The objective is to greedily select the nodes that improve the $\MSE$ while satisfying controllability. Specifically, for each candidate driving node, we check $\rank(\tbOmega)$ increases until we reach controllability as indicated in line 11 [cf. Proposition~\ref{prop_nodes_control}]. Since we are looking for the minimum energy controller, then line 12 entails computing $\bbu^{\ast} = (\tbOmega^{\Hr} \tbOmega)^{-1} \tbOmega^{\Hr} \tbx_{K}^{\ast}$ if $\rank(\tbOmega) = Tm$, and \eqref{eqn_optimal_wtau_unbiased} if $\rank(\tbOmega)=K$. While this constrained greedy approach has no theoretical guarantees \cite{Chamon18-Greedy}, our numerical results in Section~\ref{sec_sims} show that Algorithm~\ref{algm_greedy} exhibits a performance close to the optimal solution.

%%%%%%%%%%%%%%%%%%%
%%% SUBSECTION
%%%	subsec_minmse_tau
%%%%%%%%%%%%%%%%%%%

\subsection{Biased controller} \label{subsec_minmse_tau}

When the requirement for an unbiased controller is not strict, we can leverage the bias-variance trade-off to further reduce the mean squared error of the controlled state. Given a fixed sampling set $\bbC$, the $\MSE(T)$ \eqref{eqn_mse_t} is a quadratic function on the control signals $\{\bbu_{t},t=0,1,\ldots,T-1\}$. Therefore, we express $\bbu_{t}$ as a function of $\bbC$ as follows.

First, the derivative of $\MSE(T)$ in \eqref{eqn_mse_t} w.r.t. $\bbu_{t}$ is
% EQN: eqn_dmsedu
\begin{equation} \label{eqn_dmsedu}
	\frac{\partial \MSE(T)}{\partial \bbu_{t}} =
		- 2 \bbC \bbbeta_{t}
		+ 2 \sum_{\tau=0}^{T-1} \bbC \bbGamma_{t,\tau} \bbC^{T} \bbu_{\tau}
\end{equation}
for $t=0,\ldots,T -1$ with $\bbbeta_{t}$ and $\bbGamma_{t,\tau}$ given in \eqref{eqn_mse_t}. By defining $\bbbeta = [\bbbeta_{T-1}^{\Tr},\ldots,\bbbeta_{0}^{\Tr}]^{\Tr} \in \reals^{NT}$,
% EQN: eqn_def_delta
\begin{equation} \label{eqn_def_delta}
	\bbGamma =
		\begin{bmatrix}
			\bbGamma_{T-1,T-1} & \bbGamma_{T-1,T-2} & \cdots & \bbGamma_{T-1,0} \\
			\bbGamma_{T-2,T-1} & \bbGamma_{T-2,T-2} & \cdots & \bbGamma_{T-2,0} \\
			\vdots & \vdots & \ddots & \vdots \\
			\bbGamma_{0,T-1} & \bbGamma_{0,T-2} & \cdots & \bbGamma_{0,0}
		\end{bmatrix} \in \reals^{NT \times NT}
\end{equation}
and by setting \eqref{eqn_dmsedu} to zero, we can write
% EQN: eqn_ut_minmse
\begin{equation} \label{eqn_ut_minmse}
    \bbGamma_{C} \bbu = 
	\left( \bbI_{T} \otimes \bbC \right)
	\bbGamma
	\left( \bbI_{T} \otimes \bbC^{\Tr} \right)
	\bbu
	=
	\left( \bbI_{T} \otimes \bbC \right)
	\bbbeta
    =
    \bbbeta_{C}
\end{equation}
with $\bbGamma_{C} \in \reals^{MT \times MT}$ and $\bbbeta_{C} \in \reals^{MT \times 1}$. By construction $\bbGamma_{C}$ has rank $MT$ since $\bbC \in \ccalC_{M,N}$. Then, for a sampling matrix $\bbC$ such that $\rank(\bbGamma_{C}) \geq MT$, $\bbGamma_C$ is nonsingular and leads to the (parameterized) minimum $\MSE(T)$ control signals
% EQN: eqn_optimal_wtau_mse
\begin{equation} \label{eqn_optimal_utau_mse}
	\bbu^{\ast}_C = \bbGamma_{C}^{-1} \bbbeta_{C}.
\end{equation}

From this relation between the control signal and the sampling matrix, we consider a two-stage optimization approach \cite[Section 4.1.3]{Boyd04-Convex} to find $\bbC$ that minimizes the $\MSE(T)$. This optimization problem writes as
% ALGN: eqn_opt_mse
\begin{align}
\min_{\bbC \in \ccalC_{M,N}^{\ast}}
	& \alpha - 2 \ \bbbeta_{C}^{\Tr} \bbu_{C}^{\ast} + (\bbu_{C}^{\ast})^{\Tr} \bbGamma_{C} \bbu_{C}^{\ast}
		\label{eqn_opt_mse} \\
\textrm{s. t. }
	& \bbu^{\ast}_C = \bbGamma_{C}^{-1} \bbbeta_{C}
		\nonumber \\
	& \bbGamma_{C} = (\bbI_{T} \otimes \bbC) \bbGamma (\bbI_{T} \otimes \bbC^{\Tr})
		\nonumber \\
	& \bbbeta_{C} = (\bbI_{T} \otimes \bbC) \bbbeta.
		\nonumber
\end{align}

\noindent To deal with the non-convexity of \eqref{eqn_opt_mse}, similarly to \eqref{eqn_opt_unbiased}, we rely on a constrained greedy approach analogous to Algorithm~\ref{algm_greedy}. Specifically, we replace line 10 by the computation of $\bbGamma_{C}$ and $\bbbeta_{C}$, line 11 by $\rank(\bbGamma_{C}) \geq mT$, and line 12 by \eqref{eqn_optimal_utau_mse}.

%!TEX root = control.tex

%%% FIGURE %%%
%%
\begin{figure*}
    \captionsetup[subfigure]{justification=centering}
    \centering
    \begin{subfigure}{0.9\columnwidth}
        \centering
        \includegraphics[width=0.99\textwidth]{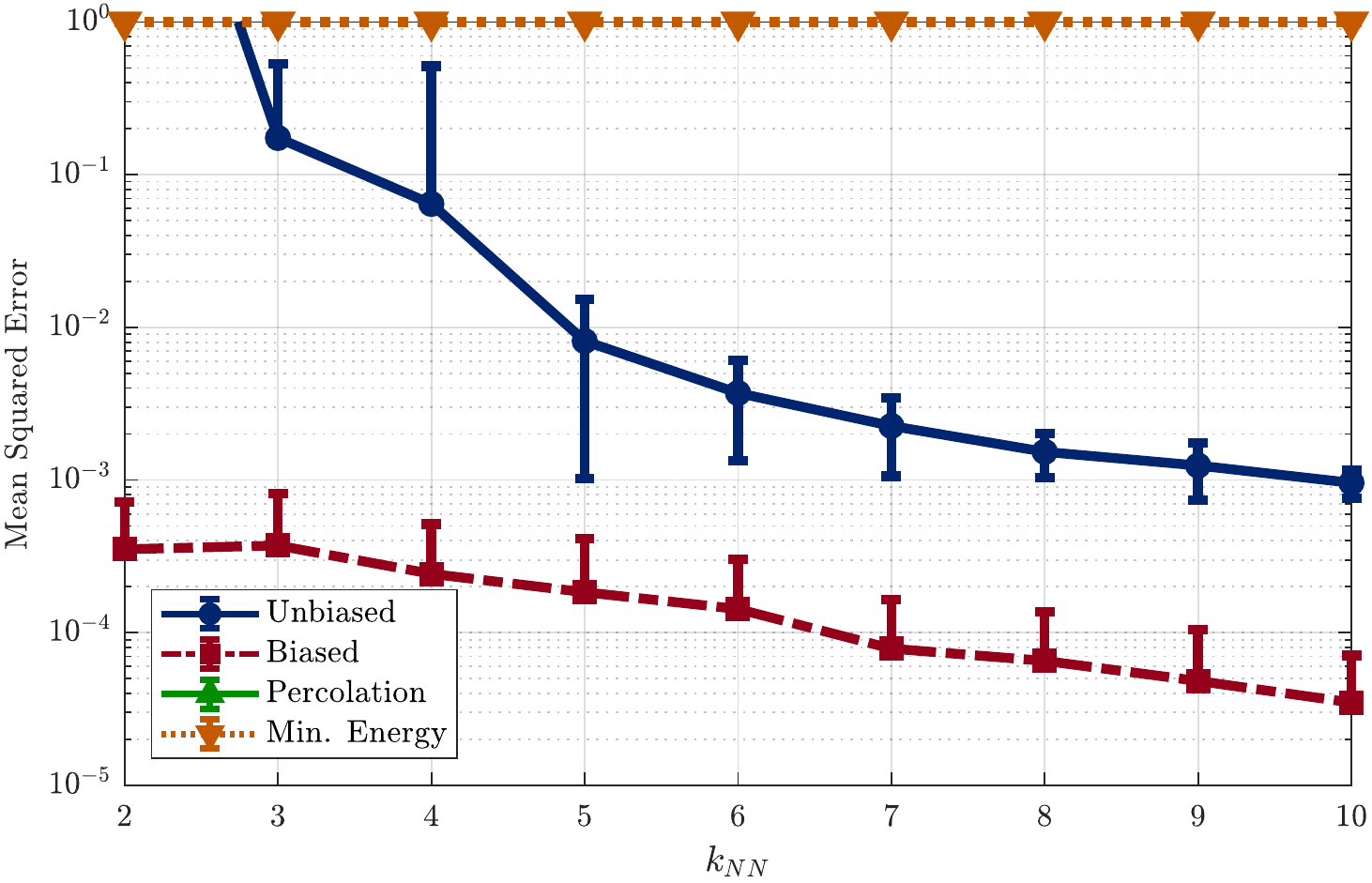}
        \caption{Geometric}
        \label{connectivityGeom}
    \end{subfigure}
    \hfill
    \begin{subfigure}{0.9\columnwidth}
        \centering
        \includegraphics[width=0.99\textwidth]{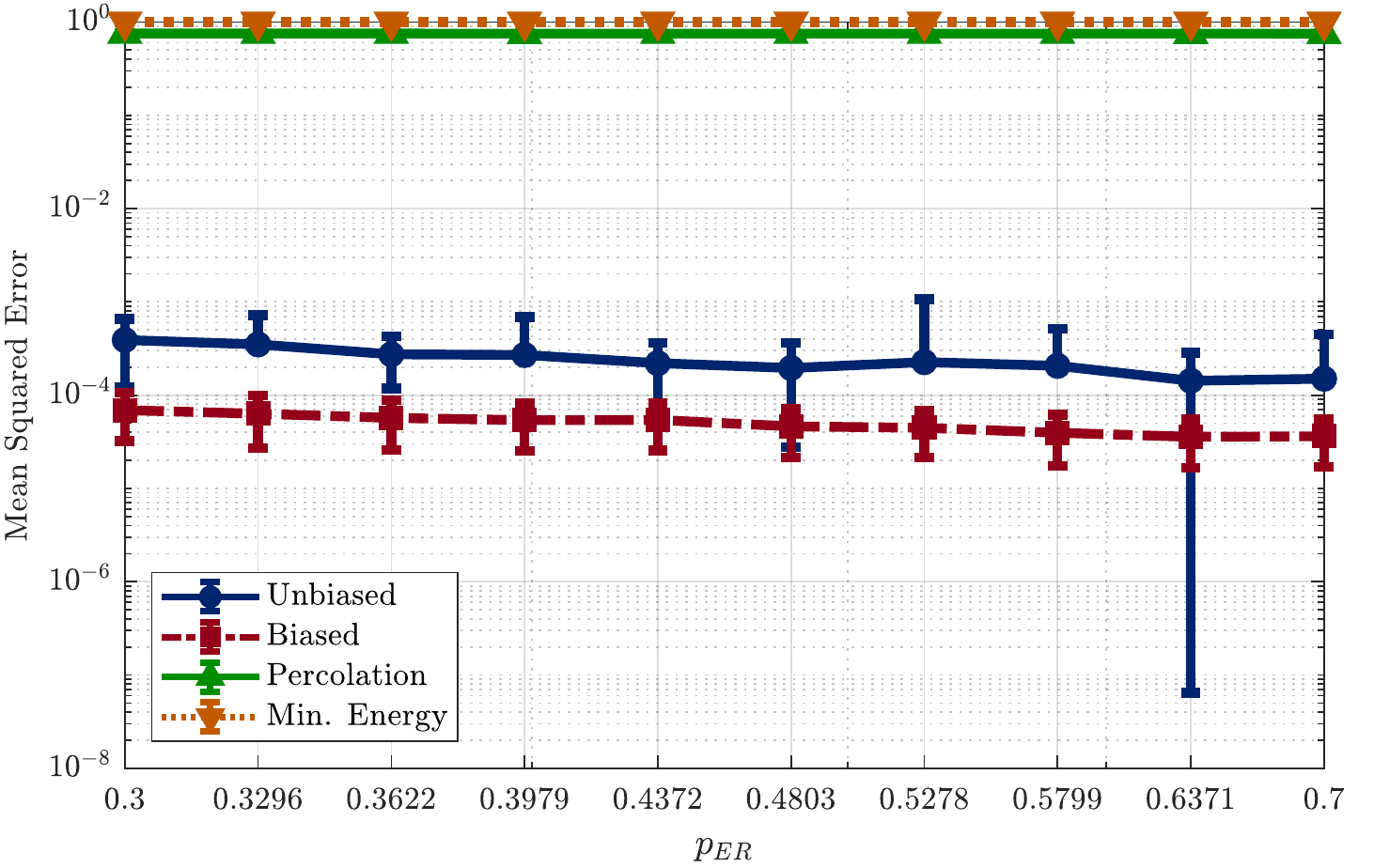}
        \caption{Erd\H{o}s-R\'{e}nyi}
        \label{connectivityER}
    \end{subfigure}
    \caption{Impact of graph connectivity measured by the average graph degree. \subref{connectivityGeom} Geometric graph: note that when the average degree increases, the connectivity is higher, and as such there are more communication paths through which the signal can flow, and thus is less affected by link losses. \subref{connectivityER} ER graph: the increase in connectivity does not lead to noticeable changes in the MSE since these graphs already have a large average degree (around $30$). The error bars indicate $3\times$ the estimated standard deviation from the $500$ graph realizations. We note that the $y$-axis limit has been set to $\MSE=1$ and that the Percolation method yields $\MSE > 1$ and therefore is not shown in \subref{connectivityGeom}.}
    \label{fig:connectivity}
\end{figure*}
%%
%%%%%%%%%%%%%%

%%% FIGURE %%%
%%
\begin{figure*}
    \captionsetup[subfigure]{justification=centering}
    \centering
    \begin{subfigure}{0.9\columnwidth}
        \centering
        \includegraphics[width=0.99\textwidth]{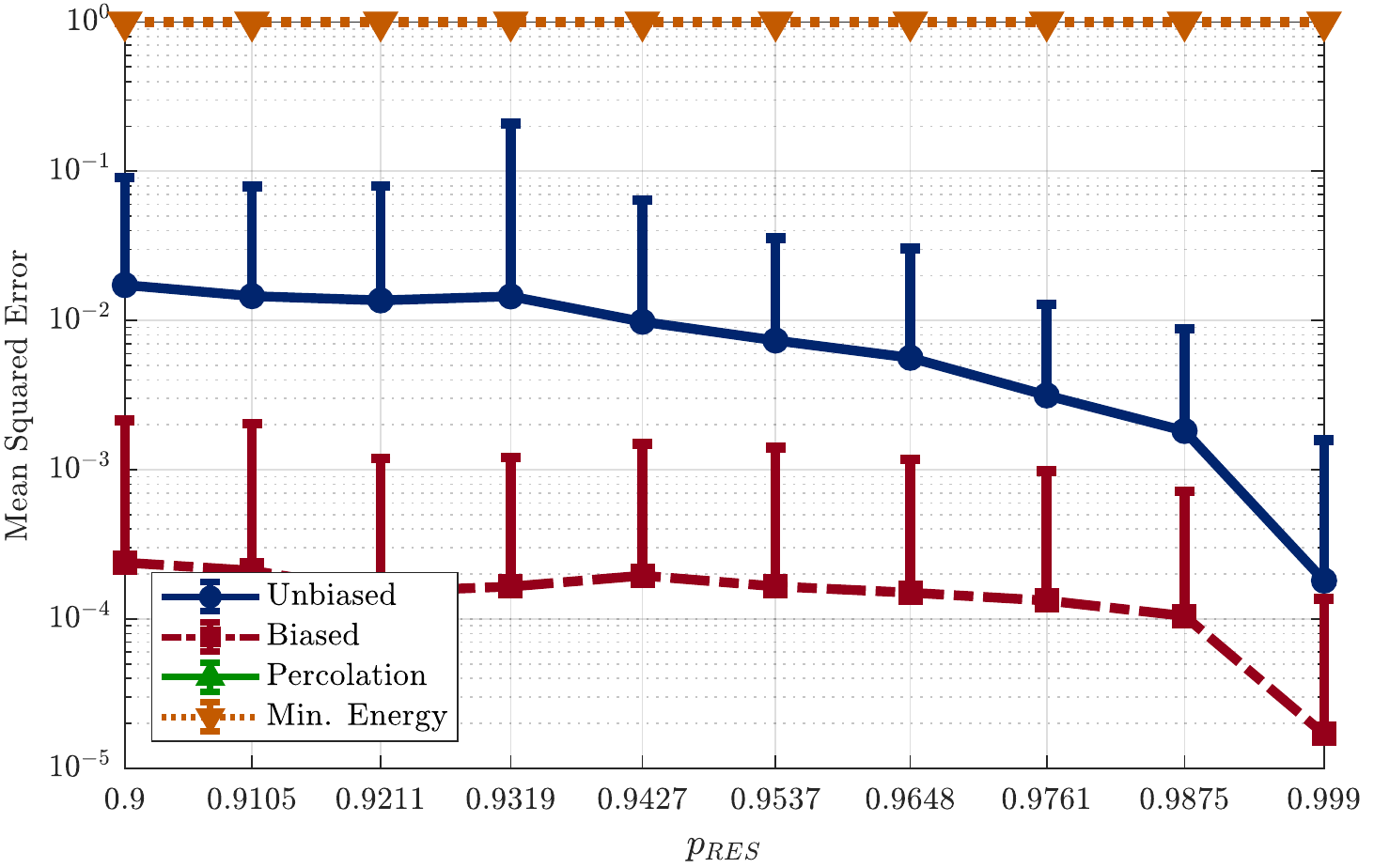}
        \caption{Geometric}
        \label{linkLossGeom}
    \end{subfigure}
    \hfill
    \begin{subfigure}{0.9\columnwidth}
        \centering
        \includegraphics[width=0.99\textwidth]{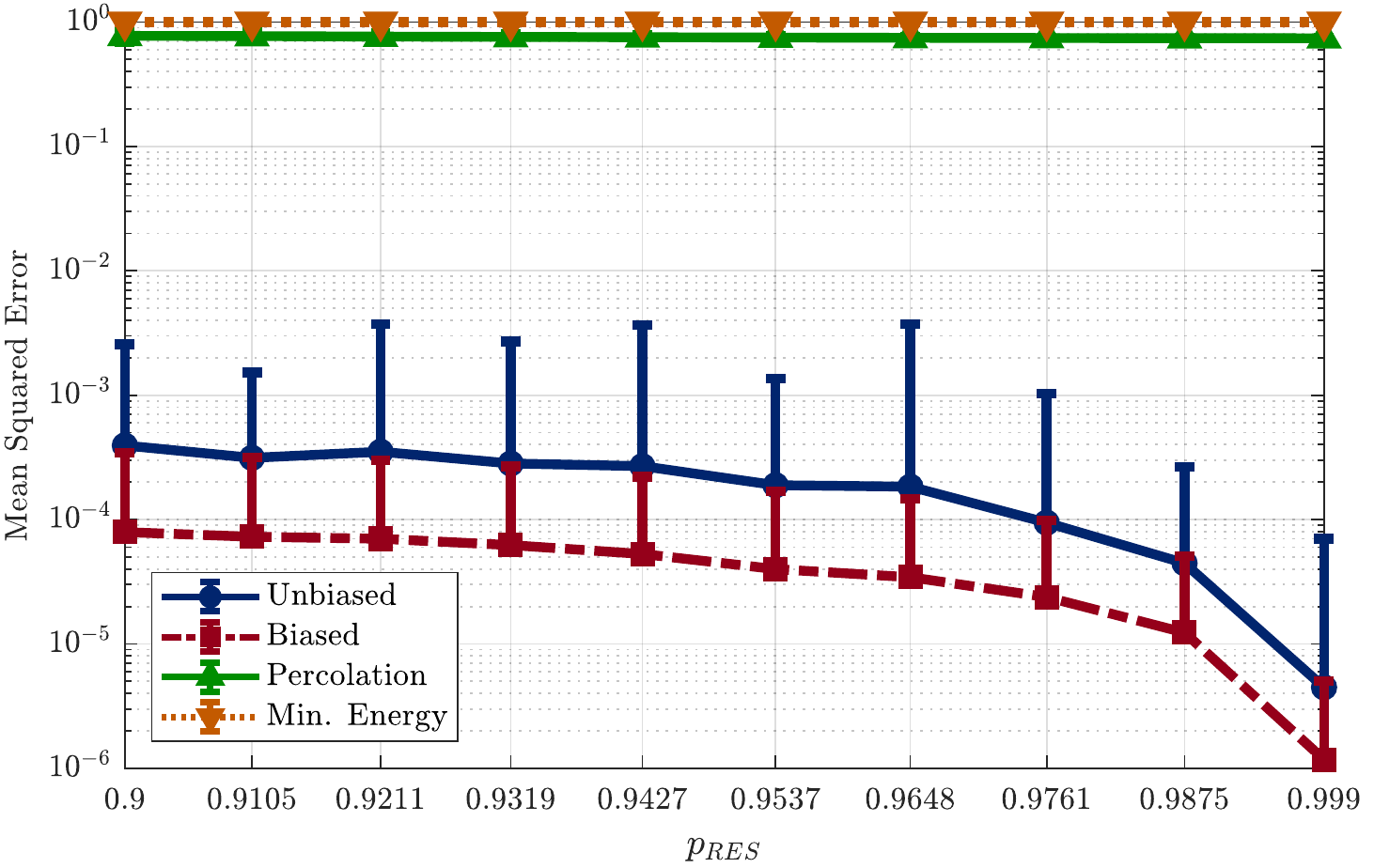}
        \caption{Erd\H{o}s-R\'{e}nyi}
        \label{linkLossER}
    \end{subfigure}
    \caption{Impact of the link loss. \subref{linkLossGeom} Geometric graph. \subref{linkLossER} ER graph. When $p_{\RES}$ increases, fewer links are lost and thus the network is easier to control, leading to a lower  MSE. The error bars indicate $3\times$ the estimated standard deviation from the $500$ graph realizations. We note that the $y$-axis limit has been set to $\MSE=1$ and that the Percolation method yields $\MSE > 1$ and therefore is not shown in \subref{linkLossGeom}.}
    \label{fig:linkLoss}
\end{figure*}
%%
%%%%%%%%%%%%%%

\section{Numerical Experiments} \label{sec_sims}

We evaluate the proposed control strategies on different scenarios to analyze the different trade-off when controlling the network. We compare the \emph{unbiased} minimal energy controller \eqref{eqn_opt_unbiased} and the \emph{biased} controller \eqref{eqn_opt_mse} with the  \emph{Percolation} control strategy of \cite{Segarra16-Percolation} and with the \emph{Min. Energy} approach of \cite{Barbarossa16-Control}. Next, we consider synthetic network models, namely Erd\H{o}s-R\'{e}nyi (ER) graphs \cite{ErdosRenyi59-RandomGraphs} and geometric graphs, while in Section~\ref{subsec_real} we test the methods on real-world social networks, namely on the Zachary's Karate Club \cite{Zachary77-KarateClub} and on a Facebook subnet \cite{McAuley12-EgoNets}.

\subsection{Synthetic network models} \label{subsec_synthetic}

The ER graph forms the edges between any two nodes randomly and independently with probability $p_{\textrm{ER}}$ and has an average degree of $p_{\textrm{ER}}N$. The geometric graph draws nodes uniformly at random in the $[0,1]^{2}$ plane and computes the Euclidean distance $d_{ij}$ between any pair of nodes. We assigned a Gaussian kernel edge weights $w_{ij} = \ccalW(v_{i},v_{j}) = e^{-d_{ij}^{2}}$ and kept only the $k_{\textrm{NN}}$ nearest neighbors per node; the parameter $k_{\textrm{NN}}$ controls the average degree. For both models, we considered realizations that result in connected graphs. To account for the randomness in the generative models and in the edge loss, we averaged the performance over $500$ different underlying graphs where for each of them we accounted also for $5000$ RES realizations.

Unless otherwise specified, we set $N = 100$, $p_{\textrm{ER}} = 0.5$, $k_{\textrm{NN}} = 5$, and the RES link loss probability to $p_{\RES}=0.95$. The control time is $T = 8$, the number of driving nodes is $M = 8$, and the initial state is $\bbx_{0}=\bbzero_{N}$. The target state $\bbx^{\ast}$ has a bandwidth of $K=10$ with GFT coefficients $\tbx_{K}^{\ast}$ decaying linearly as $[\tbx_{K}^{\ast}]_{k} = 1-(k-1)/K$ for $k=1,\ldots,K$. We measured the controllability performance between the bandlimited controlled state and target one through the normalized MSE: $\MSE(T) = \mbE[ \|\bbH \bbx_{T} - \bbx^{\ast}\|^{2}]/\|\bbx^{\ast}\|^{2}$.

\textbf{Time-invariant network.} To set a baseline, we first compared the biased \eqref{eqn_opt_mse} and the unbiased \eqref{eqn_opt_unbiased} controllers with the Percolation \cite{Segarra16-Percolation} and Min. Energy \cite{Barbarossa16-Control} strategies on a fixed time-invariant network. This is the same control scenario that is considered in \cite{Segarra16-Percolation, Barbarossa16-Control} and is equivalent to setting $p_{\RES} = 1$. We report the results for geometric graphs in Figure~\ref{fig:fixed}. Figure~\ref{fixedT} is a parametric simulation as a function of time horizon $T$ and Figure~\ref{fixedM} is a parametric simulation as a function of the number of samples. In general, we observe that the biased estimator has a performance similar to the Min. Energy and slightly better than Percolation. The unbiased controller lags behind in terms of $\MSE$.

%%% FIGURE %%%
%%
\begin{figure*}
    \captionsetup[subfigure]{justification=centering}
    \centering
    \begin{subfigure}{0.9\columnwidth}
        \centering
        \includegraphics[width=0.99\textwidth]{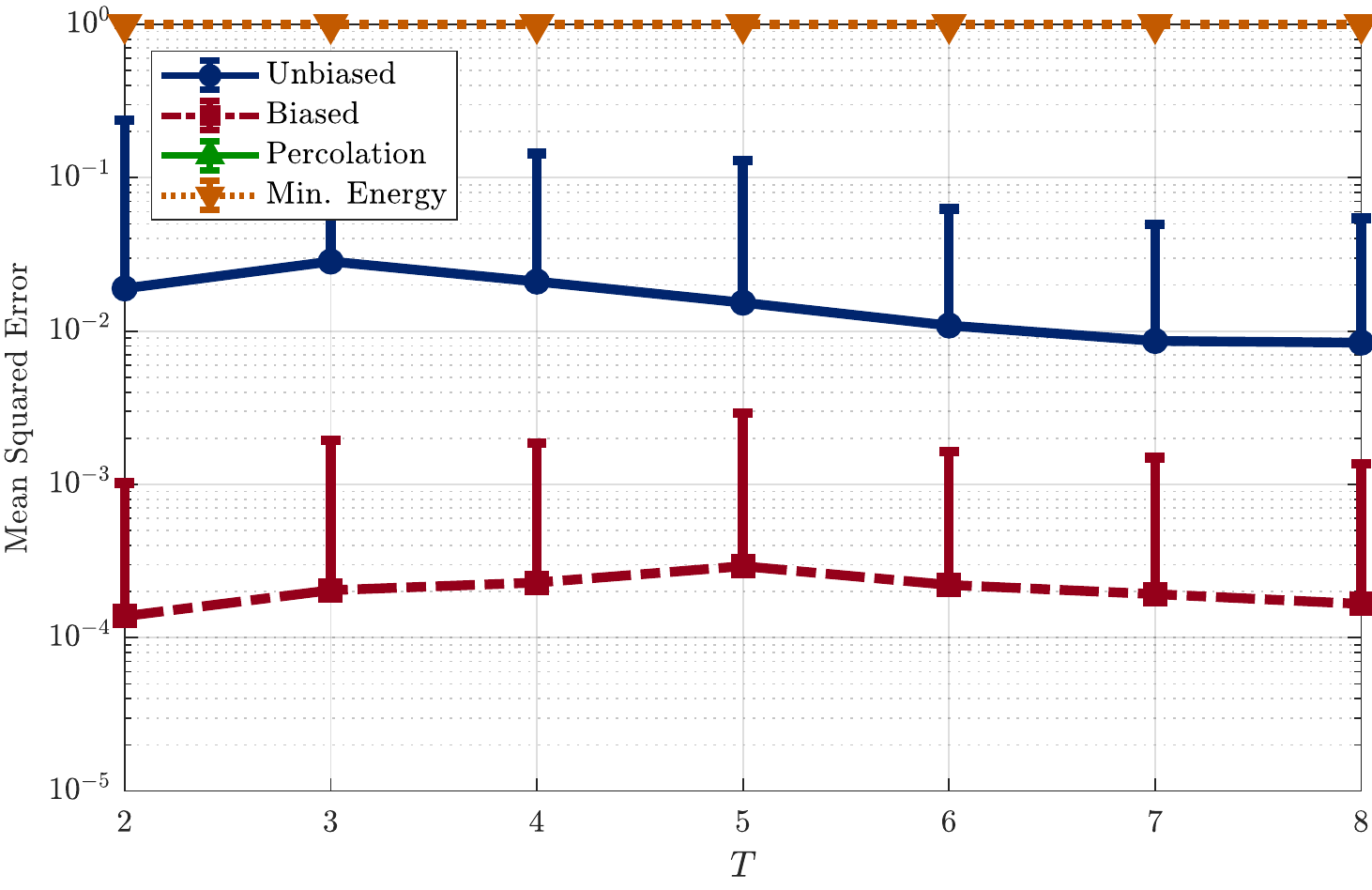}
        \caption{Geometric}
        \label{controlTimeGeom}
    \end{subfigure}
    \hfill
    \begin{subfigure}{0.9\columnwidth}
        \centering
        \includegraphics[width=0.99\textwidth]{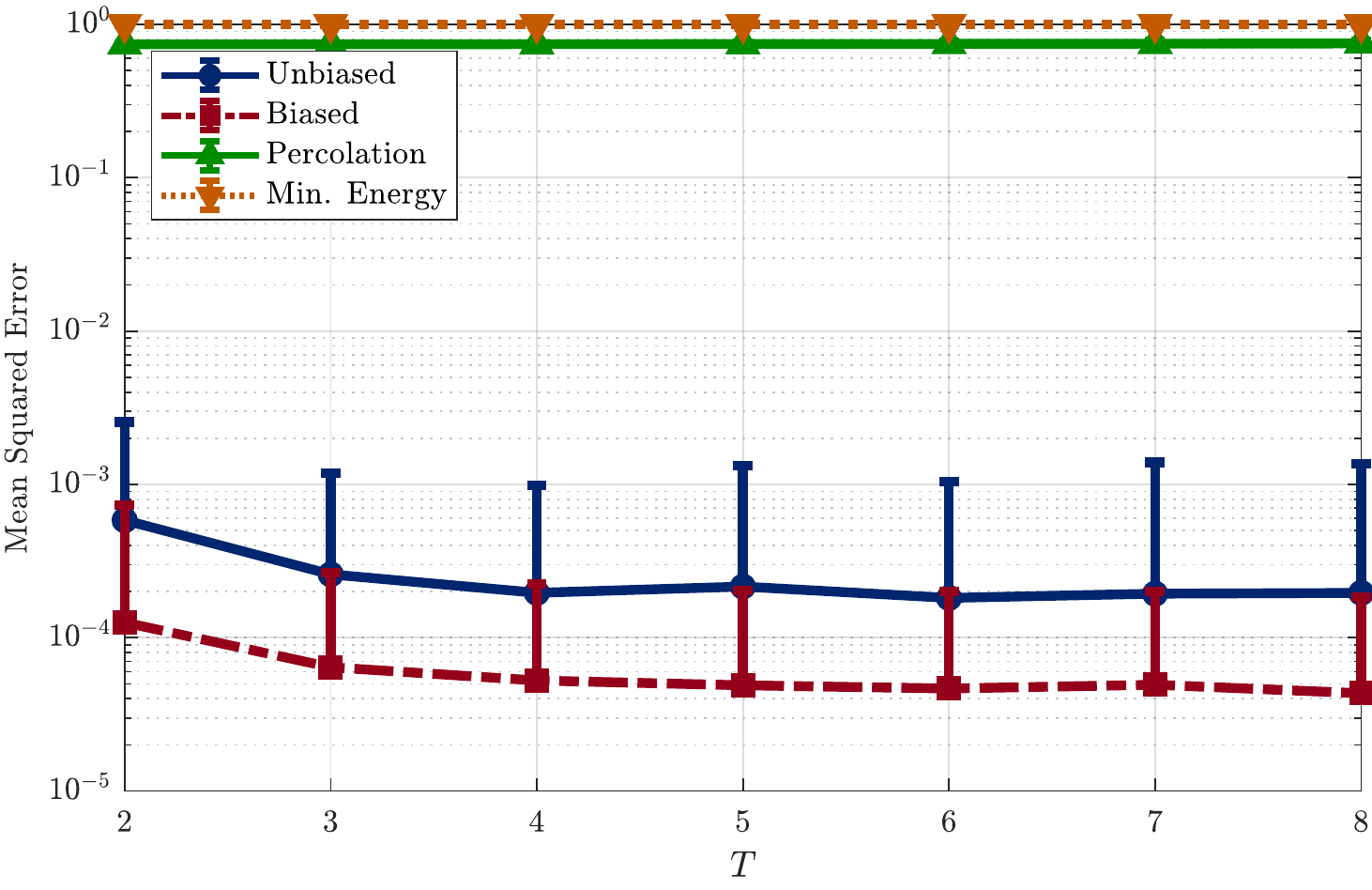}
        \caption{Erd\H{o}s-R\'{e}nyi}
        \label{controlTimeER}
    \end{subfigure}
    \caption{Impact of the control time. \subref{controlTimeGeom} Geometric graph. \subref{controlTimeER} ER graph. For the proposed strategies, an increase in control time yields a slightly lower MSE since there is more time to correct for the network evolution. The error bars indicate $3\times$ the estimated standard deviation from the $500$ graph realizations. We note that the $y$-axis limit has been set to $\MSE=1$ and that the Percolation method yields $\MSE > 1$ and therefore is not shown in \subref{controlTimeGeom}.}
    \label{fig:controlTime}
\end{figure*}
%%
%%%%%%%%%%%%%%

\textbf{Graph connectivity.} In the first random time varying experiment, we studied the impact of the graph connectivity on the controllability performance. We accounted for the graph connectivity by changing the average degrees, i.e., $p_{\textrm{ER}} N$ for the ER graph and $k_{\textrm{NN}}$ in the geometric graph. Figure~\ref{fig:connectivity} shows the MSE as the connectivity increases. From Figure~\ref{connectivityGeom}, we observe that the control on geometric networks improves with the average degree. This is intuitively satisfying since larger degrees lead to a higher connectivity between nodes; hence, it renders them more robust to the RES model. Contrarily, for the ER model in Figure~\ref{connectivityER} this behavior is not as much emphasized. We attribute this phenomenon to the large average degree of the ER graphs (above $30$) and to the relatively high value of $p_{\textrm{RES}}$. That is, the loss of a few edges does not impact the overall ability to control the network.

\textbf{Link loss.} In the second experiment, we analyzed the impact of $p_{\RES}$ for a fixed average degree. From Figure~\ref{fig:linkLoss}, we note that as $p_{\textrm{RES}}$ increases (fewer links are lost) the MSE reduces and leads to an easier to control network. This is because a higher $p_{\RES}$ yields realizations with fewer edge losses, thus more similar to the underlying (mean) graph.

\textbf{Control time.} In the third and last experiment, we analyzed the impact of the control time horizon $T$. From Figure~\ref{fig:controlTime}, we observe that the proposed strategies are not significantly affected by changes in $T$ as they only improve slightly.

From this set of experiments, we make three key observations. First, the proposed strategies offer the best performance. Second, the biased controller achieves the lowest MSE. This is expected since it levers the bias-variance trade-off to minimize the overall MSE at expenses of a bias in the controlled state. Third, not accounting for the graph randomness affects seriously the performance, even for $p_{\text{RES}} = 0.999$. In fact, the deterministic alternatives of Percolation and Min. Energy have a worse performance by orders of magnitude compared with the proposed techniques. This contrast is particularly evident when comparing with the simulations for a time-invariant network in Figure~\ref{fig:fixed}. This could be explained by the fact that losing a link has a huge impact in the topology of the graph and severely affects the eigenbasis, thereby, changing the subspace of signals that are bandlimited on a given graph.

%%% FIGURE %%%
%%
\begin{figure*}
    \captionsetup[subfigure]{justification=centering}
    \centering
    \begin{subfigure}{0.9\columnwidth}
        \centering
        \includegraphics[width=0.99\textwidth]{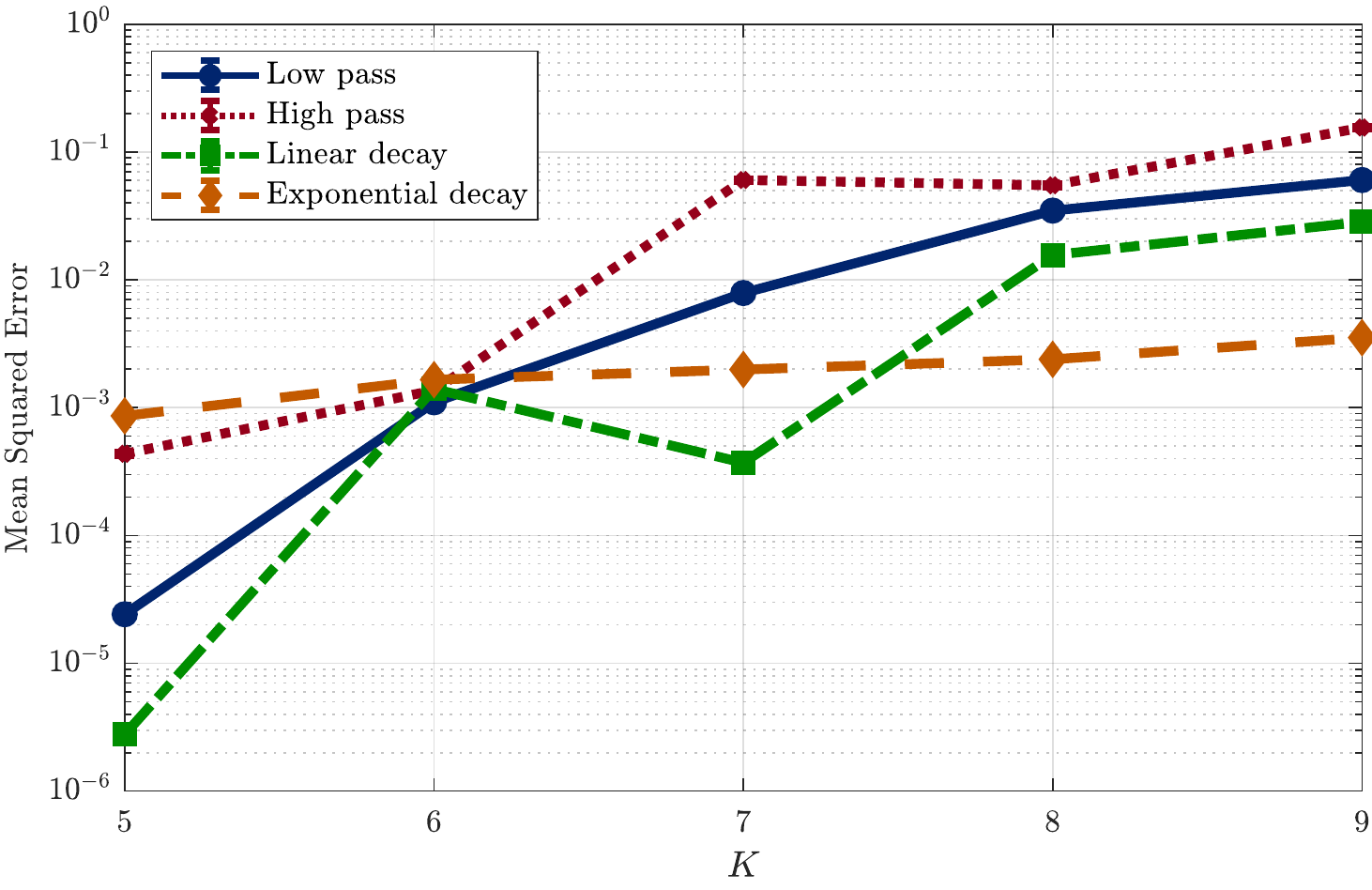}
        \caption{Zachary's Karate Club}
        \label{bandwidthZachary}
    \end{subfigure}
    \hfill
    \begin{subfigure}{0.9\columnwidth}
        \centering
        \includegraphics[width=0.99\textwidth]{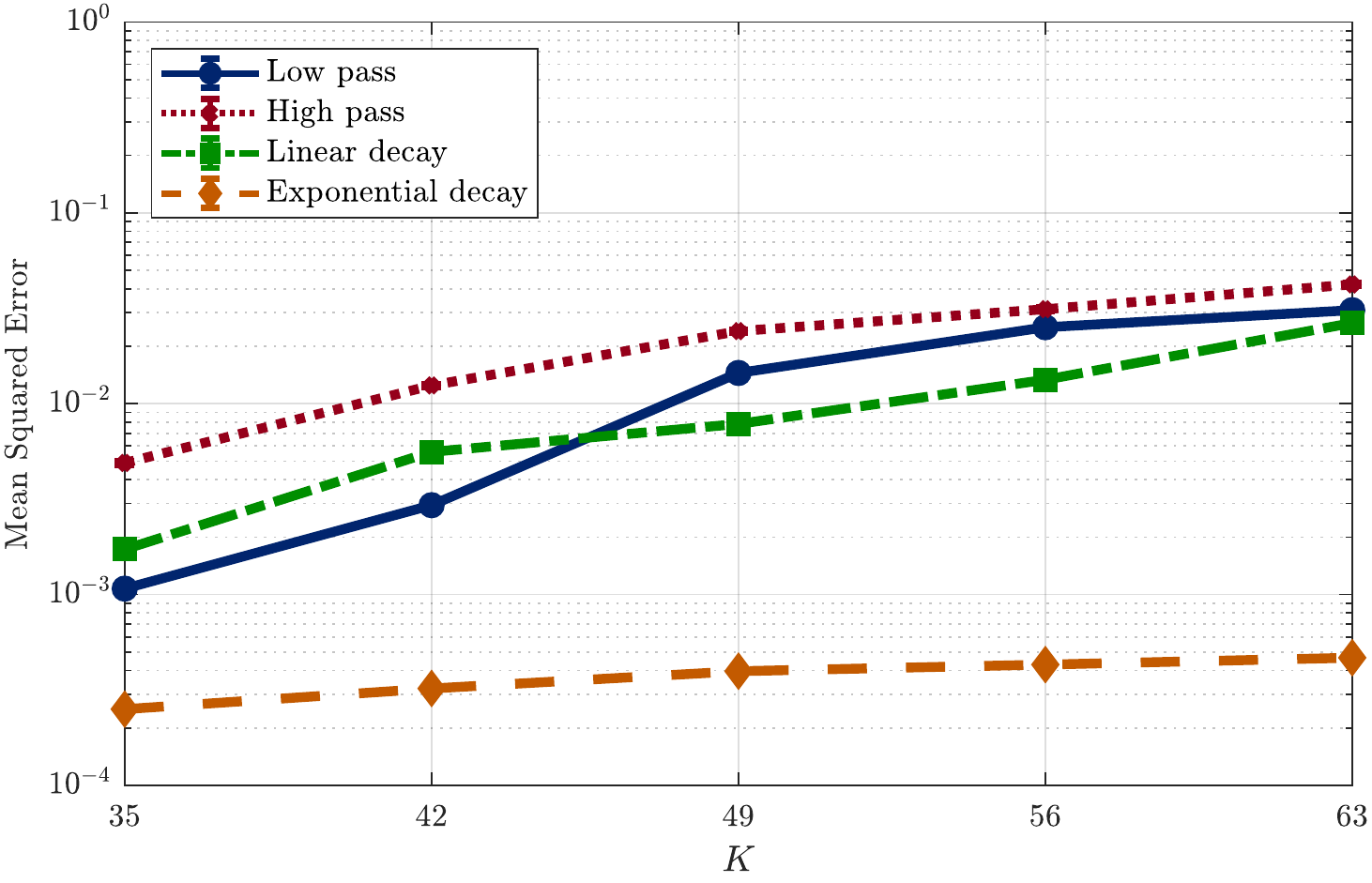}
        \caption{Facebook}
        \label{bandwidthFB}
    \end{subfigure}
    \caption{Impact of the bandwidth and the control signal shape. \subref{bandwidthZachary} Zachary's Karate Club social network. \subref{bandwidthFB} Subnet of Facebook social network. An increased bandwidth ($K$) leads to a higher MSE, since these signals are harder to control for a fixed $M$ and $T$. Higher graph frequency content signals are also harder to control.}
    \label{fig:bandwidth}
\end{figure*}
%%
%%%%%%%%%%%%%%

%%% FIGURE %%%
%%
\begin{figure*}
    \captionsetup[subfigure]{justification=centering}
    \centering
    \begin{subfigure}{0.9\columnwidth}
        \centering
        \includegraphics[width=0.99\textwidth]{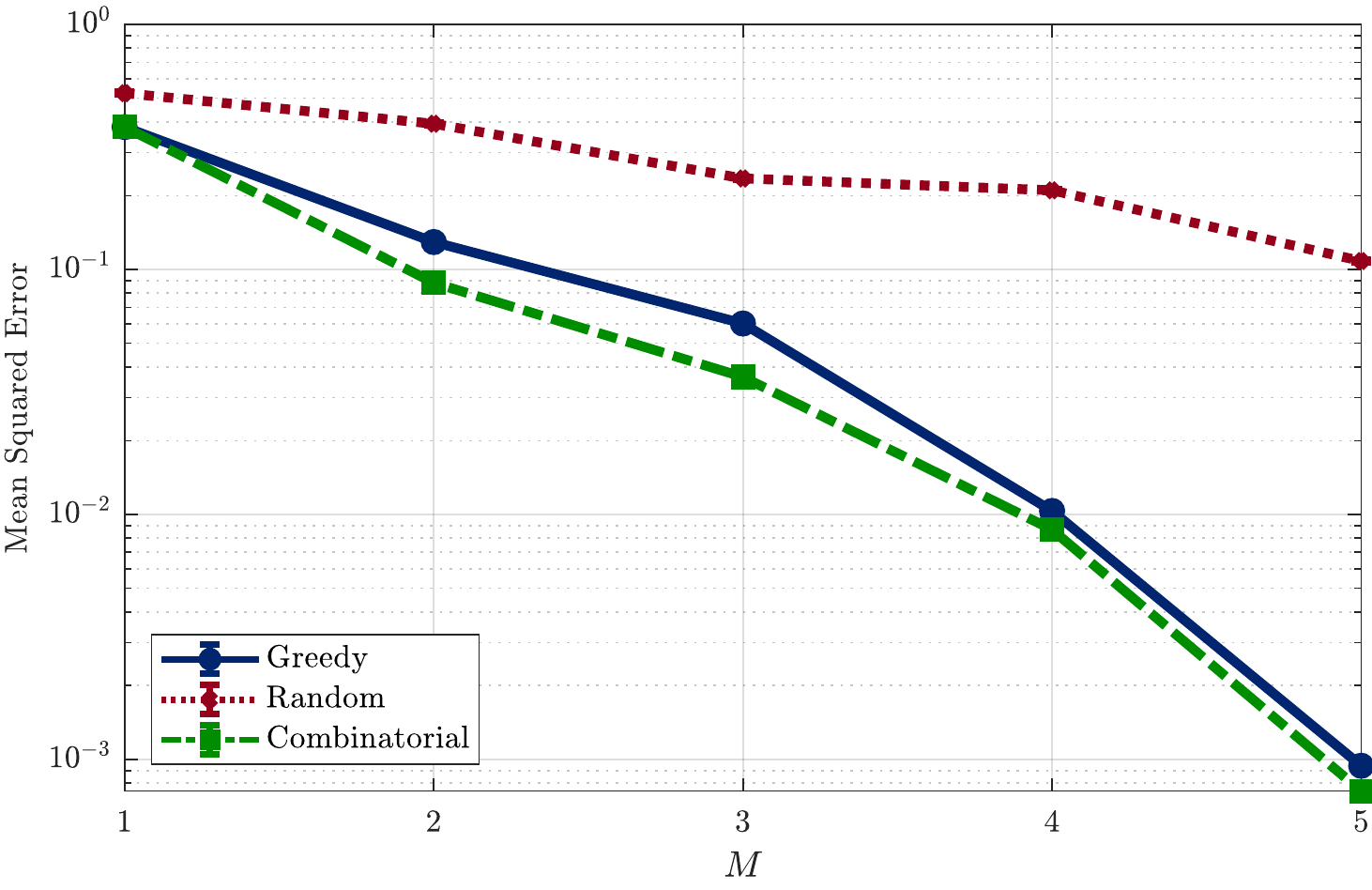}
        \caption{Zachary's Karate Club}
        \label{samplingZachary}
    \end{subfigure}
    \hfill
    \begin{subfigure}{0.9\columnwidth}
        \centering
        \includegraphics[width=0.99\textwidth]{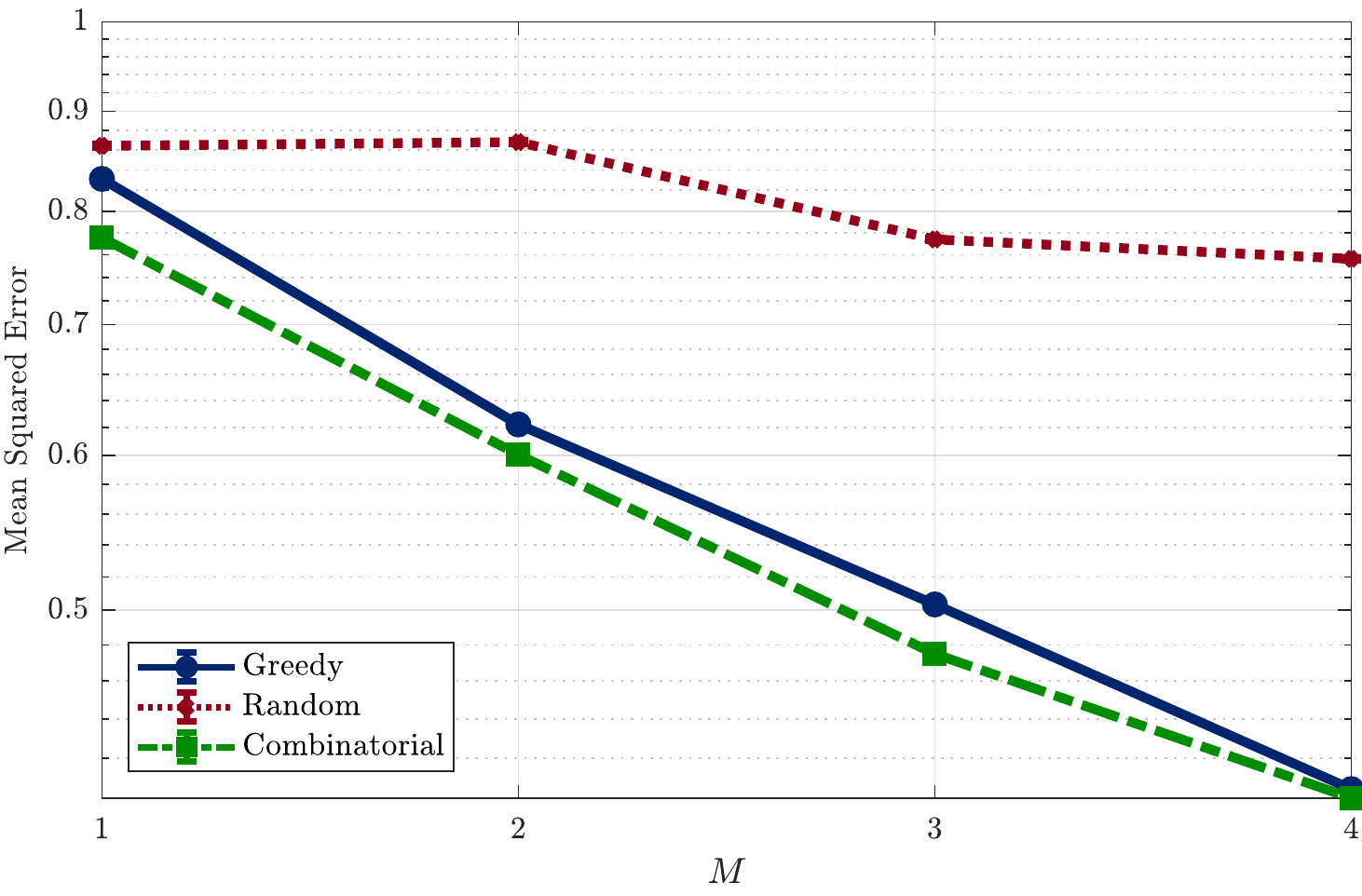}
        \caption{Facebook}
        \label{samplingFB}
    \end{subfigure}
    \caption{Impact of the number of selected nodes and sampling strategy. More control nodes lead to a smaller MSE satisfying our intuition since the degrees of freedom increase. The proposed greedy heuristic is close to the optimal combinatorial solution and represents a substantial improvement compared the random node selection.}
    \label{fig:sampling}
\end{figure*}
%%
%%%%%%%%%%%%%%

\subsection{Real world graphs} \label{subsec_real}

We consider the formation of opinion profiles on two social networks, namely the Zachary's Karate Club \cite{Zachary77-KarateClub} in Figure~\ref{fig:karateClub} of $N=34$ nodes and a Facebook subnetwork \cite{McAuley12-EgoNets} of $N = 234$ nodes. 

We set the control time to $T = 8$, the number of selected nodes to $M = \text{round}(0.08N)$, $p_{\RES} = 0.95$, and $\bbx_{0}=\bbzero_{N}$. For simplicity of presentation, we focus only on the biased controller strategy which has consistently yielded the best performance. Likewise, we do not compare it with the methods in \cite{Segarra16-Percolation, Barbarossa16-Control} given their poor performance on random time varying graphs. We again averaged the performance over $5000$ different $\RES$ realizations.

\textbf{Bandwidth and spectrum of the desired state.} In this experiment, we analyzed the impact that the desired state bandwidth and its GFT have on the controllability performance. We considered four different GFTs for the desired state, namely: (i) a step low-pass $[\tbx_{K}^{\ast}]_{k}=1$ for $k=1,\ldots,K$; (ii) a step high-pass, where the active frequencies correspond to the $K$ eigenvectors with highest total variation; (iii) a linear decay response given by $[\tbx_{K}^{\ast}]_{k} = 1 - (k-1)/K$ for $k=1,\ldots,K$; and (iv) an exponential decay response with $[\tbx_{K}^{\ast}]_{k} = e^{1-k}$ for $k=1,\ldots,K$. For a fair comparison, we normalized all desired states to unit energy and analyzed different values of $K$; a fraction of $N$ between $0.15$ and $0.27$. 

The results are depicted in Figure~\ref{fig:bandwidth}. First, we observe that controlling the system to a higher bandwidth state is harder since the set of graph frequencies to guarantee controllability increases. Second, we observe that the high-pass response is harder to achieve and responses that decay to zero (like the linear decay and the exponential decay) yield lower MSE. This is because high-pass responses are translated in the vertex domain as states having dissimilar values in adjacent nodes. They render the control of network dynamics to such states more challenging. Hence, we conclude that it is easier to drive the network state to a signal that varies smoothly on the nodes compared with a signal that has highly different values in connecting vertices; e.g., it is easier to convince someone to vote a conservative candidate, if she is surrounded by members that have the same political inclination.

\textbf{Sampling heuristics.} In the last experiment, we focused on the impact of the control nodes. We compared the proposed constrained greedy selecting heuristic in Algorithm \ref{algm_greedy} with the optimal combinatorial solution and a uniformly random sampling scheme. We fixed $K = 10$ and considered the linear decay desired state $\tbx_{K}^{\ast}$ such that $[\tbx_{K}^{\ast}]_{k} = 1 - (k-1)/K$ for $k=1,\ldots,K$. The obtained results are shown in Figure~\ref{fig:sampling}. We observe that the MSE decreases as more control nodes are selected. We also observe that the greedy heuristic yields a performance similar to the optimal solution and represents a considerable improvement over random selection.

%!TEX root = control.tex

\section{Conclusions} \label{sec_conclusions}

In this paper, we studied the problem of controlling network states. We considered a random time varying network to be driven to a desired bandlimited state. To cope with the randomness in the underlying support, we introduced the concept of controllability in the mean, where we postulated to control the system as if it were running on the expected graph. We then carried out a detailed mean squared error analysis to quantify the deviation of the target signal, when the control is designed for the expected graph but ran on any given random network realization. We used this analysis to propose two different control strategies and evaluated their performance on both synthetic graph models and real-world social networks. We concluded that it is of paramount importance to take into account the random nature of the underlying topology. We leave as future work the analysis of more complex random network models, other parsimonious graph signal models, and other control strategies that involve spectral or energetic constraints. Another direction worth investigating is the proposal of other heuristic solutions to the respective optimization problems.

%!TEX root = control.tex

%%%%%%%%%%%%%%%%
%%% Appendix %%%
%%%%%%%%%%%%%%%%

\appendices

%%%%%%%%%%%%%%%%%%%
%%% SUBSECTION
%%%	appendix_bounds
%%%%%%%%%%%%%%%%%%%

\section{Proof of Lemma~\ref{l_valid_models}.}

\begin{proof} For model $(i)$, $\bbS = \bbL$ and the system transition matrix is $\bbA_{t} = \bbI - \epsilon \bbL_{t}$ for $0 < \epsilon \leq 1/\|\bbL\|_{2}$. First, we prove that $\| \bbL_{t} \|_{2} \leq \| \bbL \|_{2} \le \varrho$. Note that $\ccalG_{t} \subseteq \ccalG$ for every $t$ and, therefore, from the Laplacian interlacing property \cite{Chen04-Interlacing} this condition always holds. The proof of Assumption 2 is straightforward, i.e., from $\mbE [ \bbA_{t} ] = \bbI - \epsilon \mbE[\bbL_{t}] = \bbI - \epsilon p \bbL$, which means that $\mbE [ \bbA_{t} ]$ and $\bbL$ share the same eigenvectors. For the last condition, note that $\|\bbA_{t}\|_{2} = \| \bbI - \epsilon \bbL_{t} \|_{2} \leq 1$ since $\epsilon \leq 1/\|\bbL\|_{2} \leq 1/\|\bbL_{t}\|_{2}$. Therefore, $\|\bbA_{t}\|_{2}$ is upper bounded by some finite $\varrho$.

For model $(ii)$, $\bbS = \bbW$ and the system transition matrix is $\bbA_{t} = \bbW_{t}$. To prove that $\|\bbW_t\|_2 \le \|\bbW\|_2$, recall that for connected graphs, the largest eigenvalue is positive and real {\cite[Theorem 0.2]{Cvetkovic79-SpectraGraphs}}. Then, since $\bbW$ is considered to be normal and Assumption \ref{ass_gso} holds, $\|\bbW\|_{2} = \lambda_{\max}(\bbW) \leq \max \deg(\ccalG)\le \varrho < \infty$. Likewise, since $\ccalG_{t} \subseteq\ccalG$, then $\max \deg(\ccalG_{t}) \leq \max \deg (\ccalG) < \infty$ and therefore $\| \bbW_{t} \|_{2}\le\varrho < \infty$ for all $t$. The proofs of the last two conditions are straightforward since $\mbE[\bbA_{t}] = p\bbW$ and $\|\bbA_{t}\|_2 = \|\bbW_{t}\|_{2} \le \varrho < \infty$. This completes the proof. \end{proof}

%%%%%%%%%%%%%%%%%%%
%%% SUBSECTION
%%%	appendix_controllability
%%%%%%%%%%%%%%%%%%%

\section{Proof of Proposition~\ref{prop.detSparCont} and Corollary~\ref{cor_suff}}

\begin{proof}[Proof of Proposition~\ref{prop.detSparCont}]
Recall that $\ccalS$ is the set of the selected $M$ nodes and that $\bbC^\Tr\bbC = \diag(\bbc)$, where $\bbc\in\{0,1\}^N$ with $[\bbc]_{i} = 1$ if $v_{i} \in \ccalS$ and $[\bbc]_{i} = 0$, otherwise. System \eqref{eqn_control_sys_freq_k} is equivalent to
% EQN: eq.equivSysDet
\begin{equation}\label{eq.equivSysDet}
	\tbx_{t,K} = \tbA_K \tbx_{t-1,K}  + \bbV_{K}^{\Hr} \diag(\bbc) \hbu_{t-1}
\end{equation}
where $\hbu_{t} \in \reals^{N \times 1}$ denotes the zero-extended control signal such that $[\hbu_{t}]_{i} =  [\bbu_{t}]_{i}$ if $v_{i} \in \ccalS$ and $[\hbu_{t}]_{i} = 0$, otherwise. Then, system \eqref{eq.equivSysDet} is controllable iff the $K\times T N$ matrix
% ALGN: eq.cont_matDet
\begin{align}\label{eq.cont_matDet}
	\tbOmega 
		&= [\bbV_K^\Hr\diag(\bbc), \tbA_K\bbV_K^\Hr\diag(\bbc), \ldots, \tbA_K^{T-1}\bbV_K^\Hr\diag(\bbc)]
			\nonumber \\
	&=[\bbI_K, \tbA_K, \ldots, \tbA_K^{T-1}] (\bbI_{T}\otimes\bbV_K^\Hr\diag(\bbc))
\end{align}
is full rank. Observe that 
% EQN: eq.rnk_cond
\begin{equation}\label{eq.rnk_cond}
% SPLIT
\begin{split}
	\rank(\tbOmega) 
		&\le \min\{K, \rank(\bbI_{T}\otimes\bbV_K^\Hr\diag(\bbc))\}
			\\
		&\le \min\{K, T\ \min\{K,M\}\}
\end{split}
\end{equation}
holds from $\rank(\bbA\bbB) \le \min\{\rank(\bbA), \rank(\bbB)\}$. Therefore, to ensure the full rank $K$ of $\tbOmega$, $M \ge \lceil K/ T\rceil$ must hold, for some $T \ge 1$. This concludes the proof.
\end{proof}

\begin{proof}[Proof of Corollary~\ref{cor_suff}]
Recall that, for two matrices $\bbX \in \reals^{M \times N}$ and $\bbY \in \reals^{N \times K}$ {\cite[Section 0.4.6]{HornJohnson85-MatrixAnalysis}}
% EQN: eqn_wikirank
\begin{equation} \label{eqn_wikirank}
\text{if}~~\rank(\bbY) = N \Rightarrow \rank(\bbX \bbY) = \rank(\bbX).
\end{equation}
The mean system \eqref{eqn_control_sys_freq_k_mean} is controllable, iff
% EQN: controllability_cor
\begin{equation} \label{eqn_controllability_cor}
\tbOmega = [\bbI_{K}, \tbA_{K},\ldots,\tbA_{K}^{T-1}] \ \left(\bbI_{T} \otimes \bbV_{K}^{\Hr} \diag(\bbc) \right)
\end{equation}
has rank $K$ with $\tbA_{K}=\diag(\barba_{K})$.

The first term in \eqref{eqn_controllability_cor} has rank
% EQN
\begin{equation}
% SPLIT
\begin{split}
	& \bbX 
	= [\bbI_{K}, \tbA_{K},\ldots,\tbA_{K}^{T-1}] \in \reals^{K \times T K} \\
	& \rank(\bbX) = K
\end{split}
\end{equation}
while the second term has rank
% EQN
\begin{equation}
% SPLIT
\begin{split}
	& \bbY = (\bbI_{T} \otimes \bbV_{K}^{\Hr} \diag(\bbc) ) \in \reals^{T K \times T N} \\
	& \rank(\bbY) = T \ \rank(\bbV_{K}^{\Hr} \diag(\bbc))
\end{split}
\end{equation}
since $\bbY$ consists of the Kronecker product of $\bbV_{K}^{\Hr}\diag(\bbc)$ with an identity matrix. Note that $\bbV_{K}^{\Hr} \diag(\bbc)$ selects indeed rows of $\bbV_{K}$.

Now, if $M \geq K$ and the node set $\ccalS$ are such that the selected $M$ rows of $\bbV_{K}$ form a set of $K$ linearly independent vectors, then $\rank(\bbV_{K}^{\Hr} \diag(\bbc)) = K$. This implies that $\rank(\bbY) = T K$ and in virtue of \eqref{eqn_wikirank}, we obtain
% EQN
\begin{equation}
\rank(\tbOmega) = \rank(\bbX \bbY) = \rank(\bbX) = K
\end{equation}
yielding that the mean system \eqref{eqn_control_sys_freq_k_mean} is controllable.
\end{proof}

%%%%%%%%%%%%%%%%%%%
%%% SUBSECTION
%%%	appendix_mse
%%%%%%%%%%%%%%%%%%%

\section{Proof of Theorem~\ref{thm_mse}}

\begin{proof} The MSE can be rewritten as
% EQN: eqn_mse_terms
\begin{equation}\label{eqn_mse_terms}
% SPLIT
\begin{split}
& \MSE(T)
	= \mbE\left[ \| \bbH \bbx_{T} - \bbx^{\ast} \|_{2} \right] 
		\\
	& \quad = \mbE \left[ \bbx_{T}^{\Tr} \bbH^{\Tr} \bbH \bbx_{T} \right] 
	 - 2 (\bbx^{\ast})^{\Tr} \bbH \mbE[\bbx_{T}] + \left\|\bbx^{\ast}\right\|_{2}^{2}.
\end{split}
\end{equation}
where each term is computed next.

First, to compute $\mbE[\bbx_{T}]$ and $\mbE[\bbx_{T}^{\Tr} \bbH^{\Tr} \bbH \bbx_{t}]$, note that $\bbx_{T}$ can be written as
% EQN: eqn_xt
\begin{equation} \label{eqn_xt}
	\bbx_{T} = 
		\sum_{\tau=0}^{T-1} \bbPhi_{T-1,\tau+1} \bbB \bbu_{\tau}
\end{equation}
where $\bbPhi_{b,a}=\bbA_{b} \bbA_{b-1} \cdots \bbA_{a+1} \bbA_{a}$ is the state transition matrix in the interval $[a, b]$ for $b > a$. Since under the $\RES(p)$ model $\bbA_{t}$ are i.i.d. matrices, $\mbE [ \bbPhi_{b,a}] = \barbA^{b-a+1}$. Thus, the expectation of \eqref{eqn_xt} is
% EQN: eqn_mean_xt
\begin{equation} \label{eqn_mean_xt}
	\bbmu_{T} = \mbE[\bbx_{T}]  
		= \sum_{\tau=0}^{T-1} \barbA^{T-\tau-1} \bbC^{\Tr} \bbu_{\tau}.
\end{equation}

For the second order moment $\mbE[\bbx_{T}^{\Tr}\bbH^{\Tr}\bbH\bbx_{T}]$, denote by $\bbQ = \bbH^{\Tr} \bbH$ and by substituting \eqref{eqn_xt} we have
% EQN: eqn_mse_second_moment_nogamma
\begin{equation} \label{eqn_mse_second_moment_nogamma}
% SPLIT
\begin{split}
& \mbE \left[ \bbx_{T}^{\Tr} \bbH^{\Tr} \bbH \bbx_{T} \right] = \mbE \left[ \bbx_{T}^{\Tr} \bbQ \bbx_{T} \right]
	 \\
& \ = \sum_{\tau=0}^{T-1} \sum_{\tau'=0}^{T-1} \bbu_{\tau}^{\Tr} \bbC \ \mbE \left[ \bbPhi_{T-1,\tau+1}^{\Tr} \bbQ \bbPhi_{T-1,\tau'+1} \right] \ \bbC^{\Tr} \bbu_{\tau}.
\end{split}
\end{equation}
Define $\bbGamma_{\tau,\tau'} = \mbE[ \bbPhi_{t-1,\tau+1} \bbQ \bbPhi_{t-1,\tau'+1}] \in \reals^{N \times N}$ [cf. \eqref{eqn_mse_t}], so that \eqref{eqn_mse_second_moment_nogamma} can be compactly written as
% ALGN: eqn_mse_second_moment
\begin{equation} \label{eqn_mse_second_moment}
\mbE \left[ \bbx_{T}^{\Tr} \bbH^{\Tr} \bbH \bbx_{T} \right]  
	=  \sum_{\tau=0}^{T-1} \sum_{\tau'=0}^{T-1} 
		\bbu_{\tau}^{\Tr} \bbC \ \bbGamma_{\tau,\tau'} \ \bbC^{\Tr} \bbu_{\tau}.
\end{equation}
Finally, by substituting \eqref{eqn_mse_second_moment} in the first term of the MSE \eqref{eqn_mse_terms} and \eqref{eqn_mean_xt} in the second term, we obtain the claimed expressions. This completes the proof.\end{proof}

%%%%%%%%%%%%%%%%%%%
%%% SUBSECTION
%%%	appendix_special
%%%%%%%%%%%%%%%%%%%

\section{Special Cases: Useful Computations of the Quadratic Term $\Gamma_{\tau,\tau'}$ in Theorem~\ref{thm_mse}}

Computation of the quadratic term $\Gamma_{\tau,\tau'}$ in Theorem~\ref{thm_mse} can turn out to be quite cumbersome for arbitrary graph shift operators $\bbS_{t}$ or transition matrices $\bbA_{t}$. In what follows, we offer two corollaries of Theorem~\ref{thm_mse} that address this issue. In particular, Corollary~\ref{cor_laplacian} gives an upper bound on the MSE that does not entail computation of second-order moments, while in Corollary~\ref{cor_adjacency} we show that, for the usually found case of undirected graphs, diffusion models in Lemma~\ref{l_valid_models} admit an exact computation. Proofs follow after the statement of the corollaries.

%%%% COROLLARY:
%% cor_laplacian
%
\begin{corollary} \label{cor_laplacian}
    Under the same conditions of Theorem~\ref{thm_mse} and from Lemma~\ref{l_valid_models}, the MSE \eqref{eqn_mse_t} can be upper bounded by
    % ALGN: eqn_mse_bound
    \begin{align}
    \MSE(T) \leq & \| \bbx^{\ast} \|_{2}^{2} - 2\sum_{\tau=0}^{T-1} (\bbx^{\ast})^{\Tr} \bbH \barbA^{T-\tau-1} \bbC^{\Tr} \bbu_{\tau} 
    \nonumber \\
    & + \sum_{\tau=0}^{T-1} \sum_{\tau'=0}^{T-1} \varrho^{2(T-\tau'+1)} \langle \bbu_{\tau'}, \bbu_{\tau} \rangle.
    \label{eqn_mse_bound}
    \end{align}
\end{corollary}

The result of Corollary~\ref{cor_laplacian} can be interpreted as the worst case scenario to account for the variability in the topology. In fact, \eqref{eqn_mse_bound} shows only first order dependence from the $\RES(p)$ model, but it does not show dependence from the second order moment. For the models in Lemma~\ref{l_valid_models}, we can consider for model $(i)$ $\varrho = 1$ ($\barbA = \bbI - \epsilon p \bbL$), while for model $(ii)$ $\varrho = \max \deg(\ccalG)$ ($\barbA = p \bbW$). Further insight then on the role of the graph variability is given by Corollary~\ref{cor_adjacency}, which shows the explicit dependence on the link activation probability $p$.

%%%% COROLLARY:
%% cor_adjacency
%
\begin{corollary} \label{cor_adjacency}
    Under the same conditions of Theorem~\ref{thm_mse} and additionally given that $\ccalG$ is an undirected graph, for the diffusion models in Lemma~\ref{l_valid_models}, the following holds:
    % EQN: eqn_gamma_ab
    \begin{equation} \label{eqn_gamma_ab}
    % SPLIT
    \begin{split}
    & \bbGamma_{\tau,\tau'} 
    = (\barbA^{\tau'-\tau})^{\Tr} \bbQ_{T-\tau'-1}
    %\mbE \left[ \bbA_{b}^{\Tr} \bbQ_{b} \bbA_{b} \right],
    \\
    & \bbQ_{a} = \mbE \left[ \bbA_{T-a}^{\Tr} \bbQ_{a-1} \bbA_{T-a} \right], \ a = 1,\ldots,T-\tau'-1
    \end{split}
    \end{equation}
    %it holds that, for $a \geq b$,
    for $\tau \leq \tau'$, $\bbQ_{0} = \bbH^{\Tr}\bbH$, and $\barbA^{\tau'-\tau} = (\mbE[\bbA_{t}])^{\tau'-\tau}$. 
    
    For model $(i)$ in Lemma~\ref{l_valid_models} ($\barbA = \bbI - \epsilon p \bbL$), $\bbQ_{a}$ is
    % ALGN: eqn_recursive_Q_lap
    \begin{align}
    & \bbQ_{a} 
    = \epsilon^{2} p^{2} \bbW^{\Tr} \bbQ_{a-1} \bbW
    \nonumber \\
    & \quad +\epsilon^{2} p(1-p) \nonumber \\
    & \qquad \qquad \cdot \left( \bbW^{\Tr} \circ \bbQ_{a-1} \circ \bbW  - \diag (\bbW^{\Tr} \circ \bbQ_{a - 1} \circ \bbW )\right)
    \nonumber \\
    & \quad + \epsilon^{2} p(1-p) \diag(\bbW^{\Tr} \diag(\bbQ_{a-1}) \bbW)
    \nonumber \\
    & \quad + 2 \epsilon p (\bbI - \epsilon p \bbD) \diag( \bbQ_{a-1} \bbW) \nonumber \\
    & \quad + \big( (\bbI - \epsilon p \bbD)^{2} + \epsilon^{2} p (1-p) \bbW^{\Tr} \bbW \big) \circ \diag(\bbQ_{a-1}).
    \label{eqn_recursive_Q_lap}
    \end{align}

    For model $(ii)$ in Lemma~\ref{l_valid_models} ($\barbA = p \bbW$), $\bbQ_{a}$ is
    % ALGN: eqn_recursive_Q
    \begin{align}
    & \bbQ_{a} 
    = p^{2} \bbW^{\Tr} \bbQ_{a-1} \bbW
    \nonumber \\
    & \quad  + p(1-p) \left( \bbW^{\Tr} \circ \bbQ_{a-1} \circ \bbW  - \diag (\bbW^{\Tr} \circ \bbQ_{a - 1} \circ \bbW )\right)
    \nonumber \\
    & \quad + p(1-p) \ \diag( \bbW^{\Tr} \diag(\bbQ_{a-1}) \bbW).
    \label{eqn_recursive_Q}
    \end{align}
\end{corollary}
%
%%
%%%

\begin{proof}[Proof of Corollary~\ref{cor_laplacian}]
From the first term in \eqref{eqn_mse_terms}, we have
% EQN: eqn_trace_quad
\begin{equation} \label{eqn_trace_quad}
\mbE \left[ \bbx_{T}^{\Tr} \bbH^{\Tr} \bbH \bbx_{T} \right]
	= \mbE \left[ \tr \left[ \bbH \bbx_{T} \bbx_{T}^{\Tr} \bbH^{\Tr} \right] \right].
\end{equation}
The trace argument in \eqref{eqn_trace_quad} can be expanded as
% EQN: eqn_quad_expr
\begin{equation} \label{eqn_quad_expr}
\bbH \bbx_{T} \bbx_{T}^{\Tr} \bbH^{\Tr}
	\!\! = \!\! \sum_{\tau=0}^{T-1} \! \sum_{\tau'=0}^{T-1}  \!
		\bbH \bbPhi_{T-1,\tau+1} \bbC^{\Tr} \bbu_{\tau} 
		\bbu_{\tau'}^{\Tr} \bbC \bbPhi_{T-1,\tau'+1}^{\Tr} \bbH^{\Tr}.
\end{equation}
From Lemma~\ref{l_valid_models}, we have $\|\bbA_{t}\|_{2} \leq \varrho$; so that, by the submultiplicativity of the spectral norm, we can write
% EQN: eqn_submult_phi
\begin{equation} \label{eqn_submult_phi}
\left\|\bbPhi_{b,a} \right\|_{2} 
	\! = \! \left\| \bbA_{b} \bbA_{b-1} \cdots \bbA_{a} \right\|_{2} 
	\! \leq \! \left\| \bbA_{b} \right\|_{2} \cdots \left\| \bbA_{a} \right\|_{2} 
	\leq \varrho^{b-a+1}.
\end{equation}
Also, observe that for any square matrix $\bbX$ and positive semidefinite matrix $\bbY$, it holds that $\tr[\bbX \bbY] \leq \|\bbX\|_{2} \tr[\bbY]$ \cite{Wang86-TraceBound}.

Given that the filter $\bbH$ does not amplify any frequency (i.e. $\|\bbH \|_{2} = 1$), then plugging back \eqref{eqn_quad_expr} and \eqref{eqn_submult_phi} into \eqref{eqn_trace_quad} yields
% ALGN: eqn_bound_quadterm
\begin{align}
& \mbE \left[ 
	\sum_{\tau=0}^{T-1} \sum_{\tau'=0}^{T-1} 
		\tr \left[ 
			\bbH \bbPhi_{T-1,\tau+1} \bbC^{\Tr} \bbu_{\tau} 
			\bbu_{\tau'}^{\Tr} \bbC \bbPhi_{T-1,\tau'+1}^{\Tr} \bbH^{\Tr} 
		\right] 
	\right]
		\nonumber \\
	& =\!\!  \sum_{\tau=0}^{T-1} \!\! \sum_{\tau'=0}^{T-1} 
		\!\mbE \left[
			\tr\! \left[ \!
				\left(\bbPhi_{T-1,\tau'+1}^{\Tr} \bbH^{\Tr} \bbH \bbPhi_{T-1,\tau+1}\right) 
				\!\!\left(\bbC^{\Tr} \bbu_{\tau} \bbu_{\tau'}^{\Tr} \bbC\right) 
			\right] 
		\right]
		\nonumber \\
	& \leq \sum_{\tau=0}^{T-1} \sum_{\tau'=0}^{T-1} 
		\mbE \left[ 
			\|\bbPhi_{T-1,\tau'+1}^{\Tr} \|_{2} \|\bbPhi_{T-1,\tau+1} \|_{2} \right] 
			\tr \left[ \bbC^{\Tr} \bbu_{\tau} \bbu_{\tau'}^{\Tr} \bbC \right]
		\nonumber \\
		& \leq \sum_{\tau=0}^{T-1} \sum_{\tau'=0}^{T-1} 
			\varrho^{2(t-\tau'+1)}
			\tr \left[ \bbC \bbC^{\Tr} \bbu_{\tau} \bbu_{\tau'}^{\Tr} \right]
		\nonumber \\
	& = \sum_{\tau=0}^{T-1} \sum_{\tau'=0}^{T-1} 
		\varrho^{2(t-\tau'+1)}
		\langle \bbu_{\tau'}, \bbu_{\tau} \rangle.
		\label{eqn_bound_quadterm}
\end{align}
Using \eqref{eqn_bound_quadterm} to bound \eqref{eqn_trace_quad} and replacing it in \eqref{eqn_mse_t} yields \eqref{eqn_mse_bound}.
\end{proof}

\begin{proof}[Proof of Corollary~\ref{cor_adjacency}]
Consider $\tau \leq \tau'$. From \eqref{eqn_mse_t}, we have
% ALGN: eqn_cor2_1
\begin{align} \label{eqn_cor2_1}
\bbGamma_{\tau,\tau'}
	&= \mbE \left[ \left( \bbA_{T-1} \cdots \bbA_{\tau+1}\right)^{\Tr} \bbQ \left( \bbA_{T-1} \cdots \bbA_{\tau'+1} \right) \right]
		\nonumber \\
	&= \mbE \left[ \bbA_{\tau+1}^{\Tr} \bbA_{\tau+2}^{\Tr} \cdots \bbA_{T-1}^{\Tr} \bbQ \bbA_{T-1} \cdots \bbA_{\tau'+2} \bbA_{\tau'+1} \right].
\end{align}
Then, since for two random matrices $\bbX,\bbY$ it holds that $\mbE[\bbX] = \mbE[\mbE[\bbX|\bbY]]$ {\cite[Theorem 34.4]{Billingsley95-ProbabilityMeasure}}, \eqref{eqn_cor2_1} becomes
% EQN
\begin{align}
	& \bbGamma_{\tau,\tau'}	\\
    & = \mbE \left[ \mbE \left[ \bbA_{\tau+1}^{\Tr} \cdot\!\!\cdot\!\!\cdot \bbA_{T-1}^{\Tr} \bbQ \bbA_{T-1} \cdot\!\!\cdot\!\!\cdot \bbA_{\tau'+1} | \bbA_{T-2},...,\bbA_{\tau+1} \right] \right]
		 \nonumber \\
	&= \mbE \left[ \bbA_{\tau+1 }^{\Tr} \cdot\!\!\cdot\!\!\cdot \mbE \left[\bbA_{T-1}^{\Tr} \bbQ \bbA_{T-1} | \bbA_{T-2},...,\bbA_{\tau+1} \right] \cdot\!\!\cdot\!\!\cdot \bbA_{\tau'+1} \right] \nonumber
\end{align}
which under the $\RES(p)$ model (i.e., matrices $\bbA_{a}$ are i.i.d.) can be written as
% EQN: eqn_cor2_indep
\begin{align} \label{eqn_cor2_indep}
&\bbGamma_{\tau,\tau'}
	= \mbE \left[ \bbA_{\tau+1}^{\Tr} \cdots \mbE \left[\bbA_{T-1}^{\Tr} \bbQ \bbA_{T-1} \right] \cdots \bbA_{\tau'+1} \right]
		 \\
	&= (\barbA^{\tau'-\tau})^{\Tr} \mbE \left[ \bbA_{\tau'+1}^{\Tr} \mbE \left[ \cdots \mbE \left[ \bbA_{T-1}^{\Tr} \bbQ \bbA_{T-1} \right] \cdots \right] \bbA_{\tau'+1} \right].
	\nonumber
\end{align}
Further, for $a \geq 1$ and assuming for now (to be proven later on) that $\bbQ_{a-1}$ is symmetric and positive semidefinite, we proceed to compute the $(i,j)$ entry of matrix $\bbQ_{a} = \mbE[\bbA_{T-a}^{\Tr} \bbQ_{a-1} \bbA_{T-a}]$. Towards this end, denote simply by $[\bbQ_{a-1}]_{ij} = q_{ij}$ and $[\bbA_{T-a}]_{ij} = a_{ij}$ for $i,j=1,\ldots,N$.

For $i \neq j$ and since $\bbA_{T-a}$ is symmetric, the $(i,j)$ element of $\bbQ_{a}$ becomes
% EQN: eqn_AQA_ij
\begin{align} \label{eqn_AQA_ij}
\mbE & \left[ [\bbA_{T-a}^{\Tr} \bbQ_{a-1} \bbA_{T-a}]_{ij} \right] 
	= \sum_{k=1}^{N} \sum_{\ell=1}^{N} \mbE[ a_{ki} a_{\ell j} ] q_{k\ell}
	 \\
	&= \sum_{\shortstack{\scriptsize $k=1; k \neq j$}}^{N} \sum_{\shortstack{\scriptsize $\ell=1; \ell \neq i$}}^{N} \mbE[ a_{ki} ] \mbE[a_{\ell j} ] q_{k\ell} + \mbE [ a_{ji} a_{ij}] q_{ji} \nonumber
\end{align}
where we have used the independence of the distinct elements in $\bbA_{T-a}$. The second term of \eqref{eqn_AQA_ij} groups the element $(i,j)$ together with $(j,i)$ due to symmetry of $\bbA_{T-a}$. Analogously, for the diagonal elements $i=j$, we get
% EQN: eqn_AQA_ii
\begin{equation} \label{eqn_AQA_ii}
% SPLIT
\begin{split}
\mbE & \left[ [\bbA_{T-a}^{\Tr} \bbQ_{a-1} \bbA_{T-a}]_{ii} \right] 
	= \sum_{k=1}^{N} \sum_{\ell=1}^{N} \mbE[ a_{ki} a_{\ell i} ] q_{k\ell}
		 \\
	&= \sum_{k=1}^{N} \sum_{\shortstack{\scriptsize $\ell=1; \ell \neq k$}}^{N} \mbE[ a_{ki} ] \mbE[a_{\ell i} ] q_{k\ell} + \sum_{k=1}^{n} \mbE [ a_{ki}^{2}] q_{kk}.
\end{split}
\end{equation}
With this in place, let us fist consider the simpler model $(ii)$ in Lemma~\ref{l_valid_models}, where $\bbS_{t} = \bbW_{t}$ and $\bbA_{t}= \bbW_{t}$. For this case, we have $[\bbA_{t}]_{ij} = B w_{ij}$, where $B$ is a Bernoulli random variable of parameter $p$ and $w_{ij} = [\bbW]_{ij}$. Then, by substituting $\mbE[a_{ij}]=p w_{ij}$ and $\mbE[a_{ij}^{2}]=(p^{2}+p(1-p))w_{ij}^{2}$ in \eqref{eqn_AQA_ij}, we get
% EQN
\begin{align}
	& p^{2} \sum_{\shortstack{\scriptsize $k=1; k \neq j$}}^{N} \sum_{\shortstack{\scriptsize $\ell=1; \ell \neq j$}}^{N} w_{ki} w_{\ell j} q_{k\ell} + \left( p^{2} + p(1-p) \right) w_{ji}^{2} q_{ji}
		 \nonumber \\
	&= p^{2} \sum_{k=1}^{N} \sum_{\ell=1}^{N} w_{ki} w_{\ell j} q_{k\ell} +  p(1-p) w_{ji}^{2} q_{ji}
\end{align}
which can be written in the compact form
% ALGN: eqn_AQA_ij_mat
\begin{align} \label{eqn_AQA_ij_mat}
\mbE & \left[ [\bbA_{T-a}^{\Tr} \bbQ_{a-1} \bbA_{T-a}]_{ij} \right] \\
	& = p^{2} [\bbW^{\Tr} \bbQ_{a-1} \bbW]_{ij} + p(1-p) [\bbW^{\Tr} \circ \bbQ_{a-1} \circ \bbW]_{ij}. \nonumber
\end{align}
Likewise, for $i=j$ \eqref{eqn_AQA_ii} becomes
% EQN:
\begin{equation}
% SPLIT
\begin{split}
	& p^{2} \sum_{k=1}^{N} \sum_{\shortstack{\scriptsize $\ell=1; \ell \neq k$}}^{N} w_{ki} w_{\ell i} q_{k\ell} + \left( p^{2} + p(1-p) \right) \sum_{k=1}^{N} w_{ki}^{2} q_{kk}
		\\
	&= p^{2} \sum_{k=1}^{N} \sum_{\ell=1}^{N} w_{ki} w_{\ell i} q_{k\ell} +  p(1-p) \sum_{k=1}^{N} w_{ki}^{2} q_{kk}
\end{split}
\end{equation}
which can also be written compactly as
% ALGN: eqn_AQA_ii_mat
\begin{align} 	\label{eqn_AQA_ii_mat}
\mbE & \left[ [\bbA_{T-a}^{\Tr} \bbQ_{a-1} \bbA_{T-a}]_{ii} \right] \\
	& = p^{2} [\bbW^{\Tr} \bbQ_{a-1} \bbW]_{ii} + p(1-p) [\bbW^{\Tr} \diag(\bbQ_{a-1}) \bbW]_{ii}. \nonumber
\end{align}

By combining \eqref{eqn_AQA_ij_mat} and \eqref{eqn_AQA_ii_mat} yields \eqref{eqn_recursive_Q}. Finally, note that if $\bbQ_{a-1}$ is symmetric and positive semidefinite, then so is $\bbQ_{a}$. To complete the proof, observe that $\bbQ_{0} = \bbH^{\Tr} \bbH$ is symmetric and positive semidefinite, thus \eqref{eqn_recursive_Q} holds for all $a \geq 1$.

For model $(i)$ in Lemma~\ref{l_valid_models}, we proceed in an analogous way. In this case, $\bbS_{t} = \bbL_{t}$ and $\bbA_{t} = \bbI - \epsilon \bbL_{t} = (\bbI - \epsilon \bbD_{t}) + \epsilon \bbW_{t}$, where $\bbD_{t} = \diag(\bbW_{t} \bbone)$ is the degree matrix. This means that $[\bbA_{t}]_{ij} = a_{ij} = \epsilon B w_{ij}$ if $i \neq j$ and $[\bbA_{t}]_{ii} = a_{ii} = 1 - \epsilon \sum_{k=1}^{N} B_{k} w_{ik}$ with $B_{k}$ being i.i.d. Bernoulli random variables with probability $p$.

Then, for $i \neq j$, from \eqref{eqn_AQA_ij} we have
% EQN: eqn_AQA_ij
\begin{align} \label{eqn_AQA_ij_prelap}
&\mbE  \left[ [\bbA_{T-a}^{\Tr} \bbQ_{a-1} \bbA_{T-a}]_{ij} \right] 
	=\\
	&= \sum_{\shortstack{\scriptsize $k=1; k \neq j, i$}}^{N} \:\, \sum_{\shortstack{\scriptsize $\ell=1; \ell \neq i;j$}}^{N} \mbE[ a_{ki} ] \mbE[a_{\ell j} ] q_{k\ell} \nonumber \\
    & \qquad \qquad \qquad  + \mbE [ a_{ji} a_{ij}] q_{ji} + \mbE [ a_{ii}] \mbE [ a_{jj} ] q_{ij}
		 \nonumber \\
	&= \epsilon^{2} p^{2} \sum_{\shortstack{\scriptsize $k=1; k \neq j, i$}}^{N} \:\, \sum_{\shortstack{\scriptsize $\ell=1; \ell \neq i,j$}}^{N} w_{ki} w_{\ell j} q_{k\ell} \nonumber
		\\
	& \qquad \qquad \qquad + \epsilon^{2} (p^{2} + p (1-p)) w_{ji}^{2} q_{ji} +(1- \epsilon p d_{i}) q_{ij}. \nonumber
\end{align}
Now, recalling that $w_{ii}=0$ (i.e., no self-loops), \eqref{eqn_AQA_ij_prelap} becomes
\begin{equation}
\begin{split}
\mbE & \left[ [\bbA_{T-a}^{\Tr} \bbQ_{a-1} \bbA_{T-a}]_{ij} \right] = \epsilon^{2} p^{2} \sum_{k=1}^{N} \sum_{\ell=1}^{N} w_{ki} w_{\ell j} q_{k\ell} 
		 \\&\qquad+ \epsilon^{2} p (1-p) w_{ji}^{2} q_{ji}  +(1- \epsilon p d_{i}) q_{ij},
\end{split}
\end{equation}
which can be further written in the compact form
% ALGN: eqn_AQA_ij_lap
\begin{align} \label{eqn_AQA_ij_lap}
\mbE & \left[ [\bbA_{T-a}^{\Tr} \bbQ_{a-1} \bbA_{T-a}]_{ij} \right] = \epsilon^{2} p^{2} [\bbW^{\Tr} \bbQ_{a-1} \bbW]_{ij}\\
 &\quad+ \epsilon^{2} p(1-p) [\bbW^{\Tr} \circ \bbQ_{a-1} \circ \bbW]_{ij} + [(\bbI - \epsilon p \bbD) \bbQ_{a-1}]_{ij} \nonumber
\end{align}

For $i=j$, we start with \eqref{eqn_AQA_ii}
% EQN:
\begin{align} 
\mbE & \left[ [\bbA_{T-a}^{\Tr} \bbQ_{a-1} \bbA_{T-a}]_{ii} \right] \\
& = \sum_{\shortstack{\scriptsize $k=1; k \neq i$}}^{N} \sum_{\shortstack{\scriptsize $\ell=1, \ell \neq k, i$}}^{N} \mbE[ a_{ki} ] \mbE[a_{\ell i} ] q_{k\ell} \nonumber \\
& + \sum_{\shortstack{\scriptsize $k=1; k \neq i$}}^{N}\!\!\!\!\!\! \mbE [ a_{ki}^{2}] q_{kk}  + 2 \mbE[ a_{ii} ] \sum_{\shortstack{\scriptsize $\ell=1; \ell \neq i$}}^{N}\!\!\!\!\!\! \mbE[a_{\ell i} ] q_{i\ell} + \mbE [ a_{ii}^{2}] q_{ii}. \nonumber
\end{align}
Then, recalling that $w_{ii}=0$, we replace the first and second order moments for each $a_{ij}$ and obtain
% EQN:
\begin{align} 
\mbE & \left[ [\bbA_{T-a}^{\Tr} \bbQ_{T-a} \bbA_{T-a}]_{ii} \right] = \epsilon^{2} p^{2} \sum_{k=1}^{N} \sum_{\shortstack{\scriptsize $\ell=1; \ell \neq k$}}^{N} w_{ki}w_{\ell i} q_{k\ell} \nonumber \\
& + \epsilon^{2} (p^{2} + p (1-p)) \sum_{k=1}^{N} w_{ki}^{2} q_{kk}  + 2 (1-\epsilon p d_{i}) \epsilon p \sum_{\ell=1}^{N} w_{\ell i} q_{i\ell} 
	\nonumber \\
	& + q_{ii} \Bigg( 1 - 2 \epsilon p d_{i} + \epsilon^{2} p^{2} \sum_{k=1}^{N} \sum_{\ell=1}^{N} w_{ik} w_{i\ell} \\
    & \qquad \qquad \qquad \qquad \qquad \qquad + \epsilon^{2} p (1-p) \sum_{k=1}^{N} w_{ik}^{2} \Bigg) . \nonumber
\end{align}
Finally, this can be rewritten as
% ALGN: eqn_AQA_ii_lap
\begin{align} 	\label{eqn_AQA_ii_lap}
\mbE & \left[ [\bbA_{T-a}^{\Tr} \bbQ_{a-1} \bbA_{T-a}]_{ii} \right] \\
	& = \epsilon^{2} p^{2} [\bbW^{\Tr} \bbQ_{a-1} \bbW]_{ii} + \epsilon^{2} p(1-p) [\bbW^{\Tr} \diag(\bbQ_{a-1}) \bbW]_{ii} \nonumber \\
	& \quad + 2 \epsilon p [(\bbI - \epsilon p \bbD) \diag( \bbQ_{a-1} \bbW)]_{ii} \nonumber \\
	& \quad + \left[\big( (\bbI - \epsilon p \bbD)^{2} + \epsilon^{2} p (1-p) \bbW^{\Tr} \bbW \big) \circ \diag(\bbQ_{a-1}) \right]_{ii}
\end{align}
completing the proof.
\end{proof}

\bibliographystyle{IEEEtran}
\bibliography{myIEEEabrv,bib-control}

\end{document}